\documentclass[12pt,reqno]{amsart}
\usepackage{amssymb}
\usepackage{amsfonts}
\usepackage{indentfirst}
\usepackage{amsopn}
\usepackage{amsmath}
\usepackage{amsthm}
\usepackage{amscd}
\usepackage{graphicx}
\usepackage[dvips]{epsfig}
\usepackage{color,xcolor}
\usepackage{subcaption}
\makeatletter\@addtoreset {equation}{section}\makeatother

\newtheorem{proposition}{Proposition}
\newtheorem{definition}{Definition}

\newtheorem{remark}{Remark}

\begin{document}

\title[Modulational instability of periodic standing waves]{\bf Modulational instability of periodic standing waves \\
in the derivative NLS equation}

\author{Jinbing Chen}
\address[J. Chen]{School of Mathematics, Southeast University, Nanjing, Jiangsu 210096, P.R. China} 
\email{cjb@seu.edu.cn}
	
\author{Dmitry E. Pelinovsky}
\address[D.E. Pelinovsky, Corresponding author]{
	Department of Mathematics, McMaster University, Hamilton, Ontario, Canada, L8S 4K1 } 
\email{dmpeli@math.mcmaster.ca}

\author{Jeremy Upsal}
\address[J. Upsal]{Department of Applied Mathematics,	University of Washington,
Seattle, WA 98195, USA} 
\email{jupsal@uw.edu}


\begin{abstract}
	We consider the periodic standing waves in the derivative nonlinear Schr\"{o}dinger (DNLS) equation arising in plasma physics. By using a newly developed algebraic method with two eigenvalues, we classify all periodic standing waves in terms of eight eigenvalues of the Kaup--Newell spectral problem located at the end points of the spectral bands outside the real line. The analytical work is complemented with the numerical approximation of the spectral bands, this enables us to fully characterize the modulational instability of the periodic standing waves in the DNLS equation.
\end{abstract}

\date{\today}
\maketitle

\section{Introduction}

The derivative nonlinear Schr\"{o}dinger (DNLS) equation arises in a long-wave,
weakly nonlinear limit from the one-dimensional compressible magnetohydrodynamic
equations in the presence of the Hall effect \cite{Mio,M}. This equation is a canonical model for Alfv\'{e}n waves propagating along the constant magnetic field in cold plasmas. It was shown by D. Kaup and A. Newell in \cite{KN-1978} that this equation has the same isospectral property as in the canonical Korteweg-de Vries (KdV) equation considered by P. Lax in \cite{Lax-1968}. For future reference, we take the DNLS equation in the following normalized form
\begin{equation}
i u_t + u_{xx} + i (|u|^2 u)_x = 0,
\label{dnls}
\end{equation}
where $i = \sqrt{-1}$ and $u(x,t) : \mathbb{R}\times \mathbb{R}\mapsto \mathbb{C}$. The DNLS equation is the compatibility condition for the following Lax pair of linear equations:
\begin{equation}\label{lax-1-intro}
\varphi_x = \left(\begin{array}{cc}
-i \lambda^2 & \lambda u\\
-\lambda \bar{u} & i \lambda^2\\
\end{array}
\right) \varphi,
\end{equation}
and
\begin{equation}\label{lax-2-intro}
\varphi_t = \left(\begin{array}{cc}
-2 i \lambda^4 + i \lambda^2 |u|^2 & 2 \lambda^3 u + \lambda (i u_x - |u|^2 u)\\
-2 \lambda^3 \bar{u} + \lambda (i \bar{u}_x + |u|^2 \bar{u}) & 2 i \lambda^4 - i \lambda^2 |u|^2\\
\end{array}
\right) \varphi,
\end{equation}
where $\bar{u}$ denotes the complex-conjugate of $u$ and $\varphi(x,t) : \mathbb{R}\times \mathbb{R} \mapsto \mathbb{C}^2$. The $x$-derivative part (\ref{lax-1-intro}) of the Lax pair is referred to as the Kaup--Newell (KN) spectral problem.

When the DNLS equation is posed on the real line, the Cauchy problem is locally well-posed in $H^s(\mathbb{R})$ for $s \geq \frac{1}{2}$ \cite{Takaoka} 
and ill-posed in $H^s(\mathbb{R})$ for $s < \frac{1}{2}$ due to lack of the continuous dependence on initial data \cite{BL}. If functional-analytic methods are used, global well-posedness of the Cauchy problem can only be shown for $u \in H^s(\mathbb{R})$, $s \geq \frac{1}{2}$ with small initial data in $L^2(\mathbb{R})$ (see \cite{H-O-1,H-O-2,M-W-X} and more recent works \cite{Fukaya,Wu1,Wu2}). On the other hand, by using tools of the inverse scattering transform, one can solve the Cauchy problem globally in a subspace of $H^2(\mathbb{R})$ without restricting the $L^2(\mathbb{R})$ norm of the initial data \cite{Perry2,Perry1,Yusuke1,Yusuke2}.

Travelling solitary waves of the DNLS equation are well known due to their important applications in plasma physics \cite{KN-1978,Mio,M}. These solutions can be expressed as the standing wave
\begin{equation}
\label{standing-wave}
u(x,t) = \phi_{\omega,\nu}(x-\nu t) e^{i \omega t},
\end{equation}
where $\phi_{\omega,\nu}$ is available in the polar form
\begin{equation}
\label{polar-form}
\phi_{\omega,\nu}(x) = R_{\omega,\nu}(x) e^{i \Theta_{\omega,\nu}(x)},
\end{equation}
with 
\begin{equation}
\label{soliton}
R_{\omega,\nu}(x) = \left( \frac{2 (4 \omega - \nu^2)}{
		\sqrt{4 \omega} \; \cosh(\sqrt{4\omega - \nu^2} x) - \nu} \right)^{1/2}, 
	\quad
	\Theta_{\omega,\nu}(x) = \frac{\nu}{2} x - \frac{3}{4} \int_{-\infty}^x R_{\omega,\nu}(y)^2 dy. 
\end{equation}
The speed parameter $\nu$ is arbitrary, whereas the frequency parameter $\omega$ is restricted under the constraint $4 \omega - \nu^2 > 0$. 
Orbital stability of the travelling waves in the energy space $H^1(\mathbb{R})$ was proven for $\nu < 0$ \cite{GuoWu} and for arbitrary $\nu \in (-\sqrt{4\omega},\sqrt{4 \omega})$ \cite{CO06} (see also recent works \cite{Kwon,MiaoTangXu}).

There are very few results available on the periodic standing wave solutions, which can be expressed in the form (\ref{standing-wave}) with 
\begin{equation}
\label{periodic-wave}
|\phi_{\omega,\nu}(x+L)| = |\phi_{\omega,\nu}(x)|, \quad x \in \mathbb{R},
\end{equation} 
for some fundamental period $L > 0$. The function $\phi_{\omega,\nu}$ is generally quasi-periodic in $x$ as it is expressed 
in the polar form (\ref{polar-form}), where $R_{\omega,\nu}$ and $\Theta_{\omega,\nu}'$ are periodic in $x$.

The simplest periodic standing wave solutions to the DNLS equation (\ref{dnls}) were analyzed directly in \cite{Chow} 
by separating the variables in the polar form (\ref{polar-form}).
Convergence of periodic waves to the solitary waves (\ref{soliton}) 
was shown in \cite{Hayashi}. Spectral stability of periodic waves with non-vanishing $\phi_{\omega,c}$ was established with respect to perturbations of the same period in \cite{Hakaev}.

{\em The main purpose of this work} is to classify all periodic standing waves of the DNLS equation in the form (\ref{standing-wave}) and (\ref{polar-form})
and to characterize their spectral stability with respect to localized 
perturbations. We use an algebraic method which allows us to relate the periodic standing waves with solutions of the complex finite-dimensional Hamiltonian systems. 

The algebraic method of nonlinearization of linear equations in the Lax pair to finite-dimensional Hamiltonian systems was developed by C.W. Cao and X.G. Geng in the context of the KdV equation \cite{Cao1}. The finite-dimensional Hamiltonian systems were obtained for the DNLS equation (\ref{dnls}) in \cite{Cao2,Ma2,Qiao,Zhou}. Quasi-periodic (algebro-geometric) solutions to the DNLS equation have been analyzed using the complex  finite-dimensional Hamiltonian systems in \cite{Chen} (see also \cite{Geng1,WrightDNLS,ZhaoFan} for other studies of quasi-periodic solutions in the DNLS equation). 

In the context of the periodic standing waves, the algebraic method gives the location of particular eigenvalues of the KN spectral problem which correspond to bounded periodic eigenfunctions (see \cite{CPkdv,CPnls,CPgardner,PW} for analysis of other integrable equations). These particular eigenvalues play an important role in the study of modulational stability of the periodic standing waves \cite{Kam1,Kam2}, e.g., in the propagation of dispersive shocks from an initial step-like discontinuity \cite{Kam3}. If the Floquet spectrum is obtained in the KN spectral problem analytically or numerically, then the Floquet spectrum in the linearized DNLS equation is obtained by a simple transformation (see \cite{CPW,DS} for an application of this technique to the NLS equation). 
The Floquet spectrum here refers to the union of all admissible values of the spectral parameter, for which the corresponding eigenfunctions are bounded. 

The relation between the Floquet spectrum of the KN spectral problem 
and the modulational stability of the periodic waves in the DNLS equation and other related equations  was used in the analysis \cite{DU2}, some results 
of which will be made more precise here. For numerical computations, we use the Hill's method developed in \cite{FFHM} and rigorously justified in \cite{Curtis,JohnZumb}. 

We shall now explain {\em the main results of this work.}

\begin{enumerate}
\item The periodic standing waves of the DNLS equation are derived by using the algebraic method with two eigenvalues. This method starts with a constraint imposed on solutions $u$ of the DNLS equation (\ref{dnls}) and solutions $\varphi$ of the linear equations (\ref{lax-1-intro}) and (\ref{lax-2-intro}) for two fixed values of the spectral parameter $\lambda$, all of which are apriori unknown. It is then shown that the constraint is compatible with the DNLS equation (\ref{dnls}) if and only if $u$ is a standing wave in the form (\ref{standing-wave}) and the two fixed values of $\lambda$ are found from roots of the polynomial $P(\lambda)$ of degree eight. Parameters of the polynomial $P(\lambda)$ are uniquely selected from parameters of the standing wave solutions in the form (\ref{standing-wave}). \\

\item When $u$ is the standing periodic waves of the DNLS equation (\ref{dnls}) 
in the form (\ref{standing-wave}), eight roots of the polynomial $P(\lambda)$ are related to eight eigenvalues of the KN spectral problem (\ref{lax-1-intro}) in the space of periodic or anti-periodic boundary conditions (the squared eigenfunctions are $L$-periodic). We prove that these eigenvalues could form either four pairs of purely imaginary eigenvalues or two quadruplets of four eigenvalues in four quadrants of the complex plane or the mixture of both (two pairs of purely imaginary eigenvalues and one quadruplet of complex eigenvalues). \\

\item Performing numerical computations of the Floquet spectrum 
for the KN spectral problem (\ref{lax-1-intro}), we show that 
the spectral bands connecting the eight eigenvalues determine uniquely 
the modulational stability or instability of the periodic standing waves. 
All periodic waves are spectrally (and modulationally) unstable in the case of one or two quadruplets of complex eigenvalues, whereas they are spectrally (and modulationally) stable in the case of four pairs of purely imaginary eigenvalues. 
\end{enumerate}

The paper is written as a developing argument. The most important result 
obtained by using the algebraic method assisted with the numerical computations 
of the Floquet spectrum is the precise characterization of the spectrally stable periodic standing waves:

\vspace{0.2cm}

\centerline{\fbox{\parbox[cs]{1,0\textwidth}{
The periodic standing waves of the DNLS equation are spectrally stable if and only if the eight roots of the polynomial $P(\lambda)$ are located on the imaginary axis.
}}}

\vspace{0.2cm}

The spectrally stable periodic standing waves of the form (\ref{standing-wave}) correspond to a subset of the parameter space with 
\begin{equation}
\label{stab-wave}
4 \omega < \nu^2, \quad \omega > 0, \quad \nu > 0.
\end{equation}
Periodic standing waves for every other choices in the parameter space are 
spectrally (and modulationally) unstable.

Our results are only applicable to the periodic standing waves. 
The solitary waves of the form (\ref{soliton}) 
exist only if $4 \omega > \nu^2$, where they are stable. 
Spectral bands of the Floquet spectrum for the periodic standing waves  
outside the real and imaginary axis shrink to the quadruplet of complex 
eigenvalues as the period $L$ becomes infinite. As a result, the 
spectrally unstable periodic waves converge to the spectrally 
(and orbitally) stable solitary waves, very similarly to 
the NLS equation, where periodic standing waves are spectrally unstable 
\cite{DS} and the solitary waves are spectrally (and orbitally) stable \cite{Weinstein}.

The paper is organized as follows. Properties of eigenvalues 
of the KN spectral problem are reviewed in Section \ref{sec:1}. Section \ref{sec:2} describes the algebraic method with two eigenvalues. Eight roots of the polynomial $P(\lambda)$ characterize parameters of the periodic standing waves and coincide with the eigenvalues of the KN spectral problem. The connection between the eight eigenvalues, the Floquet spectrum, and the stability spectrum is described in Section \ref{sec:7}. Section \ref{sec:3} describes all possible periodic standing waves and relates them to the location of eight roots of the polynomial $P(\lambda)$.  Numerical results on the Floquet spectrum in the KN spectral problem and the spectral stability problem are given in Section \ref{sec:4} for the physically relevant family of the periodic standing waves. The paper is concluded with Section \ref{sec:5} where we describe open directions of this study.

\section{Properties of eigenvalues for the DNLS equation}
\label{sec:1}

The following definition of eigenvalues is used in what follows.

\begin{definition}
	\label{def-eigenvalue}
	Assume that $u(x) = R(x) e^{i\Theta(x)}$ with $L$-periodic $R$ and $\Theta'$. Then $\lambda$ is called an eigenvalue of the KN spectral 
	problem (\ref{lax-1-intro}) w.r.t. periodic (anti-periodic) boundary conditions if a nonzero eigenvector $\varphi = (p,q)^T$ is given by $p(x) = P(x) e^{i\Theta(x)/2}$ and $q(x) = Q(x) e^{-i\Theta(x)/2}$ with $L$-periodic ($L$-anti-periodic) $P$ and $Q$.
\end{definition}

\begin{remark}
	If $P^2$ and $Q^2$ are $L$-periodic, then $P$ and $Q$ could be either $L$-periodic or $L$-anti-periodic.
\end{remark}

The following two propositions describe two symmetries of eigenvalues in the spectral problem (\ref{lax-1-intro}) related to the DNLS equation (\ref{dnls}).

\begin{proposition}
	\label{prop-1}
	Assume $\lambda \in \mathbb{C} \backslash \mathbb{R}$
	is an eigenvalue of the spectral problem (\ref{lax-1-intro})
	with the eigenvector $\varphi = (p,q)^T$. Then,
	$\bar{\lambda}$ is also an eigenvalue with the eigenvector
	$\varphi = (\bar{q},-\bar{p})^T$. If $\lambda \in \mathbb{R}\backslash \{0\}$ is an eigenvalue,
	then it is at least double with two eigenvectors $\varphi = (p,q)^T$
	and $\varphi = (\bar{q},-\bar{p})^T$.
\end{proposition}

\begin{proof}
	If $\varphi = (p,q)^T$ satisfies (\ref{lax-1-intro}), then 
	\begin{eqnarray}
	p_x = -i \lambda^2 p + \lambda u q, \qquad 
	q_x = -\lambda \bar{u} p + i \lambda^2 q.
	\end{eqnarray}
	Taking the complex-conjugate equation, we verify that $\varphi = (\bar{q},-\bar{p})^T$ satisfies the same equation (\ref{lax-1-intro}) but with $\lambda$ replaced by $\bar{\lambda}$. Since the periodicity properties for $\varphi = (p,q)^T$ and $\varphi = (\bar{q},-\bar{p})^T$ are the same as in Definition \ref{def-eigenvalue}, if $\lambda$ is an eigenvalue, then $\bar{\lambda}$ is an eigenvalue. 
	
	If $\lambda \in \mathbb{C}\backslash \mathbb{R}$, then $\bar{\lambda} \neq \lambda$. If $\lambda \in \mathbb{R}$, then $\bar{\lambda} = \lambda$ but $\varphi = (\bar{q},-\bar{p})^T$ is linearly independent from $\varphi = (p,q)^T$, because if there is a nonzero constant $c \in \mathbb{C}$ such that $\bar{q} = c p$ and $-\bar{p} = c q$, then $|c|^2 = -1$, a contradiction. Hence, $\lambda$ is at least a double eigenvalue.
\end{proof}

\begin{proposition}
	\label{prop-2}
	Assume $\lambda \in i \mathbb{R} \backslash \{0\}$
	is a simple eigenvalue of the spectral problem  (\ref{lax-1-intro})
	with the eigenvector $\varphi = (p,q)^T$. Then there is $c \in \mathbb{C}$ with $|c| = 1$ such that $p = c \bar{q}$.
\end{proposition}

\begin{proof}
	If $\varphi = (p,q)^T$ satisfies (\ref{lax-1-intro}) with $\lambda = i \beta$, $\beta \in \mathbb{R}$, then 
	\begin{eqnarray}
	p_x = i \beta^2 p + i \beta u q, \qquad 
	q_x = -i \beta \bar{u} p - i \beta^2 q.
	\end{eqnarray}
	Taking the complex-conjugate equation, we verify that $\varphi = (\bar{q},\bar{p})^T$ satisfies the same equation (\ref{lax-1-intro}) with the same $\lambda = i \beta$. Since $\lambda = i \beta$ is a simple eigenvalue, then $\varphi = (\bar{q},\bar{p})^T$ is linearly dependent on $\varphi = (p,q)^T$, so that there is a nonzero constant $c \in \mathbb{C}$ such that $\bar{q} = c p$ and $\bar{p} = c q$. The two relations yield the constraint $|c|^2 = 1$.
\end{proof}

\begin{remark}
	The symmetry of eigenvalues and eigenvectors in Propositions \ref{prop-1} and \ref{prop-2} holds for the second Lax equation (\ref{lax-2-intro}).
\end{remark}

\section{Algebraic method with two eigenvalues}
\label{sec:2}

In order to develop the algebraic method, it is natural to extend the DNLS equation (\ref{dnls}) as a reduction $v = \bar{u}$ of the following coupled system:
\begin{equation}
\left\{ \begin{array}{l}
i u_t + u_{xx} + i (u^2 v)_x = 0,  \\
-i v_t + v_{xx} - i (u v^2)_x = 0.
\end{array} \right.
\label{dnls-coupled}
\end{equation}
The coupled DNLS system (\ref{dnls-coupled}) appears as a compatibility condition of the Lax pair of linear equations on $\varphi \in \mathbb{C}^2$ given by
\begin{equation}\label{lax-1}
\varphi_x = U \varphi,\quad
U = \left(\begin{array}{cc}
-i \lambda^2 & \lambda u\\
-\lambda v & i \lambda^2\\
\end{array}
\right),
\end{equation}
and
\begin{equation}\label{lax-2}
\varphi_t = V \varphi, \quad
V = \left(\begin{array}{cc}
-2 i \lambda^4 + i \lambda^2 uv & 2 \lambda^3 u + \lambda (i u_x - u^2 v)\\
-2 \lambda^3 v + \lambda (i v_x + u v^2) & 2 i \lambda^4 - i \lambda^2 uv\\
\end{array}
\right).
\end{equation}
We only consider the KN spectral problem (\ref{lax-1}) and ignore the time-dependent equation (\ref{lax-2}) for now. As a result, we replace partial derivatives in $x$ with ordinary derivatives.

\begin{remark}
The time evolution of constraints in the algebraic method for the periodic standing waves is trivial (see, e.g., \cite{CPkdv,CPnls} for other integrable equations). The time evolution of the eigenvector $\varphi = (p,q)^T$ satisfying (\ref{lax-2}) is defined in (\ref{varphi-separation}) below.
\end{remark}

Fix two values $\lambda_1, \lambda_2 \in \mathbb{C}$. 
Let $\varphi = (p_1,q_1)^T$ be a particular solution to the spectral problem (\ref{lax-1}) for $\lambda = \lambda_1$ and 
$\varphi = (p_2,q_2)^T$ be a particular solution to the spectral problem (\ref{lax-1}) for $\lambda = \lambda_2$. The two solutions are required to be linearly independent if $\lambda_1 = \lambda_2$. As an ansatz, we set the following constraint between the potentials $(u,v)$ and the squared eigenfunctions:
\begin{equation}\label{2.1}
\left\{ \begin{array}{l}
u = \lambda_1 p_1^2 + \lambda_2 p_2^2, \\
v = \lambda_1 q_1^2 + \lambda_2 q_2^2. \end{array}
\right.
\end{equation}

\begin{remark}
All the above are unknowns: $\lambda_1$, $\lambda_2$, $\varphi = (p_1,q_1)^T$, $\varphi = (p_2,q_2)^T$, $u$, and $v$. 
The purpose of the algebraic method is to identify 
the unknowns from the constraint (\ref{2.1}).
\end{remark}

With the constraints (\ref{2.1}), the spectral problem
(\ref{lax-1}) for $\lambda = \lambda_1$ and $\lambda = \lambda_2$ can be written as the complex Hamiltonian system
\begin{eqnarray*}
%\label{2.3}
\frac{d p_1}{dx} = -i \lambda_1^2 p_1 + \lambda_1 (\lambda_1 p_1^2 + \lambda_2 p_2^2) q_1 = -\frac{\partial H}{\partial q_1}, \;\; &&
\frac{d q_1}{dx} = i \lambda_1^2 q_1 - \lambda_1 (\lambda_1 q_1^2 + \lambda_2 q_2^2) p_1 = \frac{\partial H}{\partial p_1},\\
%\label{2.3a}
\frac{d p_2}{dx} = -i \lambda_2^2 p_2 + \lambda_2 (\lambda_1 p_1^2 + \lambda_2 p_2^2) q_2 = -\frac{\partial H}{\partial q_2}, \;\; &&
\frac{d q_2}{dx} = i \lambda_2^2 q_2 - \lambda_2 (\lambda_1 q_1^2 + \lambda_2 q_2^2) p_2 = \frac{\partial H}{\partial p_2},
\end{eqnarray*}
generated by the complex-valued Hamiltonian
\begin{equation}\label{2.4}
H= i \lambda_1^2 p_1q_1 + i \lambda_2^2 p_2 q_2 - \frac{1}{2} (\lambda_1p_1^2 + \lambda_2 p_2^2) (\lambda_1 q_1^2 + \lambda_2 q_2^2).
\end{equation}
The complex Hamiltonian system 
admits another complex conserved quantity $F$ given by
\begin{equation}\label{2.7}
F = i (p_1q_1 + p_2 q_2),
\end{equation}
where the normalization factor $i$ is used for convenience.

If $v = \bar{u}$, we need to restrict the eigenvalues $\lambda_1$ and $\lambda_2$
in order to ensure that the conserved quantities $H$ and $F$ are real-valued. This is done in agreement with the symmetries in Propositions \ref{prop-1} and \ref{prop-2}.

\begin{itemize}
\item Let $\lambda_1 \in \mathbb{C}\backslash i \mathbb{R}$ and set
$\lambda_2 = \bar{\lambda}_1$. By Proposition \ref{prop-1}, 
we may take
	\begin{equation}
	\label{reduction-2}
	p_2 = \bar{q}_1, \quad q_2 = -\bar{p}_1.
	\end{equation}
Under the choice (\ref{reduction-2}), 
the constraint (\ref{2.1}) becomes 
compatible with the complex-conjugate symmetry
\begin{equation}
\label{constraint-red-2}
\left\{ \begin{array}{l}
u = \lambda_1 p_1^2 + \bar{\lambda}_1 \bar{q}_1^2, \\
\bar{u} = \lambda_1 q_1^2 + \bar{\lambda}_1 \bar{p}_1^2, \end{array}
\right.
\end{equation}
whereas $H$ and $F$ in (\ref{2.4}) and (\ref{2.7}) become real-valued:
\begin{equation}
	\label{Hamiltonian-2}
	H = i (\lambda_1^2 p_1q_1- \bar{\lambda}_1^2 \bar{p}_1 \bar{q}_1) - \frac12 \left|\lambda_1p_1^2 + \bar{\lambda}_1 \bar{q}_1^2 \right|^2
\end{equation}
and
\begin{equation}
\label{conserved-2}
F= i (p_1q_1-\bar{p}_1 \bar{q}_1).
\end{equation}
	
\item Let $\lambda_1,\lambda_2 \in i \mathbb{R}$ such that 
$\lambda_1 \neq \pm \lambda_2$. By Proposition \ref{prop-2}, 
we may take
	\begin{equation}
	\label{reduction-1}
\lambda_1 = i \beta_1,	\quad q_1 = -i\bar{p}_1 \quad \mbox{\rm and} \quad
\lambda_2 = i \beta_2, \quad q_2 = -i\bar{p}_2.
		\end{equation}
		Under the choice (\ref{reduction-1}), the constraint (\ref{2.1}) becomes compatible with the complex-conjugate symmetry
\begin{equation}\label{constraint-red-1}
\left\{ \begin{array}{l}
u = i \beta_1 p_1^2 + i \beta_2 p_2^2, \\
\bar{u} =  -i \beta_1 \bar{p}_1^2 - i \beta_2 \bar{p}_2^2, \end{array}
\right.
\end{equation}
whereas $H$ and $F$ in (\ref{2.4}) and (\ref{2.7}) become real-valued:
	\begin{equation}
	\label{Hamiltonian-1}
	H =  -\beta_1^2 |p_1|^2 - \beta_2^2 |p_2|^2  - \frac12 \left| \beta_1 p_1^2 + \beta_2 p_2^2 \right|^2
	\end{equation}
	and
\begin{equation}
\label{conserved-1}
F= |p_1|^2 + |p_2|^2.
\end{equation}	
\end{itemize}

\vspace{0.25cm}

Let us now derive and integrate the differential equations
on $(u,v)$ from compatibility of the constraint (\ref{2.1})
with the Hamiltonian system generated by the Hamiltonian (\ref{2.4}). From now on, we use the complex-conjugate reduction $v = \bar{u}$ in all subsequent computations, hence, we only use the constraints (\ref{constraint-red-2}) and (\ref{constraint-red-1}).

\begin{proposition}
	\label{prop-tech-1}
If $u$ satisfies either the constraint (\ref{constraint-red-2}) if $\lambda_2 = \bar{\lambda}_1$ with ${\rm Re}(\lambda_1) \neq 0$ or 
the constraint (\ref{constraint-red-1}) if $\lambda_1, \lambda_2 \in i \mathbb{R}$ with $\lambda_1 \neq \pm \lambda_2$, then $u$ is a solution 
of the second-order differential equation
\begin{equation}\label{2.9}
\frac{d^2 u}{d x^2} + i \frac{d}{dx} (|u|^2u) + 2i c \frac{d u}{d x} - 4 b u = 0,
\end{equation}
with parameters
\begin{equation}\label{2.10}
b = \lambda_1^2 \lambda_2^2 (1 + F), \qquad
c = \lambda_1^2 + \lambda_2^2 + H.
\end{equation}
\end{proposition}

\begin{proof}
By taking one derivative of either (\ref{constraint-red-2}) or (\ref{constraint-red-1}) and using (\ref{2.4}), we obtain the first-order equation on $u$,
\begin{equation}\label{2.8}
\frac{d u}{d x} + i |u|^2u + 2i H u + 2i (\lambda_1^3 p_1^2 + \lambda_2^3 p_2^2) = 0.
\end{equation}
The first-order equation (\ref{2.8}) is not closed on $u$.
However, by taking another derivative of (\ref{2.8}),
using (\ref{2.4}), (\ref{2.7}), and (\ref{2.8}), we obtain the closed second-order differential equation (\ref{2.9}) with parameters given by 
(\ref{2.10}).
\end{proof}

\begin{remark}
The second-order equation (\ref{2.9}) arises in the standing wave reduction
of the DNLS equation (\ref{dnls}) for the solutions of the form 
\begin{equation}
\label{trav-wave}
u(x,t) = \tilde{u}(x+2ct)e^{4ibt},
\end{equation}
where $\tilde{u}$ satisfies (\ref{2.9}) with tilde notations dropped.
\end{remark}

It follows from \cite{Chen} that the complex Hamiltonian system generated by the Hamiltonian (\ref{2.4}) is equivalent to the Lax equation
\begin{equation}\label{2.11}
\frac{d}{dx} \Psi = [\mathcal{U},\Psi],
\end{equation}
where 
\begin{equation}
\label{matrix-U}
\mathcal{U} = \left(\begin{array}{cc}
	-i \lambda^2 & \lambda (\lambda_1 p_1^2 + \lambda_2 p_2^2) \\
	-\lambda (\lambda_1 q_1^2 + \lambda_2 q_2^2) & i \lambda^2\\
\end{array}
\right),
\end{equation}
and
\begin{equation}\label{2.12}
\Psi := \left(\begin{array}{cc}
\Psi_{11} & \Psi_{12} \\
\Psi_{21} & -\Psi_{11}
\end{array} \right)
\end{equation}
with 
\begin{eqnarray}
\Psi_{11} &=& - i - \frac{\lambda_1^2 p_1 q_1}{\lambda^2 - \lambda_1^2} -
\frac{\lambda_2^2 p_2 q_2} {\lambda^2- \lambda_2^2}, \label{2.13}\\
\Psi_{12} &=& \lambda\left(
\frac{\lambda_1 p_1^2}{\lambda^2 - \lambda_1^2} + \frac{\lambda_2 p_2^2} {\lambda^2-\lambda_2^2}
\right), \label{2.14}\\
\Psi_{21} &=& -\lambda\left(
\frac{\lambda_1 q_1^2}{\lambda^2 - \lambda_1^2} + \frac{\lambda_2 q_2^2} {\lambda^2-\lambda_2^2}
\right). \label{2.15}
\end{eqnarray}
Here $\lambda \in \mathbb{C}$ is arbitrary spectral parameter 
and $\lambda_1$, $\lambda_2$, $(p_1,q_1)$, and $(p_2,q_2)$ are the same as in the constraints (\ref{2.1}).
It follows from (\ref{2.4}), (\ref{2.7}), and (\ref{2.13})--(\ref{2.15}) that
\begin{equation}
\label{2.19}
\det \Psi = - \Psi_{11}^2 - \Psi_{12} \Psi_{21} = 1 -
\frac{2H \lambda^2- \lambda_1^2 \lambda_2^2 F (F+2)}
{(\lambda^2 - \lambda_1^2)(\lambda^2 - \lambda_2^2)}.
\end{equation}

\begin{remark}
	\label{remark-5}
It follows from (\ref{2.19}) that if $\lambda_1 \neq \pm \lambda_2$, then $\det \Psi$ only contains simple poles at $(\pm \lambda_1,\pm \lambda_2)$ with the residue terms being independent of $x$.
\end{remark}

\begin{proposition}
	\label{prop-tech-2}
	If $u$ satisfies either the constraint (\ref{constraint-red-2}) if $\lambda_2 = \bar{\lambda}_1$ with ${\rm Re}(\lambda_1) \neq 0$ or 
	the constraint (\ref{constraint-red-1}) if $\lambda_1, \lambda_2 \in i \mathbb{R}$ with $\lambda_1 \neq \pm \lambda_2$, then $\lambda_1$ and $\lambda_2$ must be chosen from the eight roots of the polynomial
	\begin{equation}
	\label{polynomial}
	P(\lambda) = \lambda^8 -2c \lambda^6 +  (a+2b+c^2) \lambda^4 + (d-c(a+2b)) \lambda^2 + b^2,
	\end{equation}
	where parameters $b$ and $c$ are given by (\ref{2.10}) and 
	parameters $a$ and $d$ are given by 
\begin{equation}\label{3.21}
a = \lambda_1^2 \lambda_2^2 F^2 - H^2,\qquad d= \lambda_1^2 \lambda_2^2 F H (F +2) - H^2(\lambda_1^2+ \lambda_2^2 + H).
\end{equation}
\end{proposition}

\begin{proof}
By using (\ref{2.1})--(\ref{2.7}), (\ref{2.10}), and (\ref{2.8}), 
the entries of the Lax matrix $\Psi$ can be rewritten in terms of $(u,\bar{u})$ by
\begin{eqnarray}
\label{2.16}
\Psi_{11} &=& \frac{-i}{(\lambda^2 - \lambda_1^2)(\lambda^2-\lambda_2^2)} \left[\lambda^4 - \lambda^2 \left(c+ \frac12\left|u\right|^2\right) +b
\right], \\
\label{2.17}
\Psi_{12} &=& \frac{\lambda}{(\lambda^2 - \lambda_1^2)(\lambda^2-\lambda_2^2)}
 \left[\lambda^2 u + \frac{i}{2} \frac{d u}{d x}-\frac{1}{2} u\left| u\right|^2 -c u \right], \\
 \Psi_{21} &=& \frac{-\lambda}{(\lambda^2 - \lambda_1^2)(\lambda^2-\lambda_2^2)}
 \left[\lambda^2 \bar{u} - \frac{i}{2} \frac{d \bar{u}}{d x}-\frac{1}{2} \bar{u} \left| u\right|^2 - c \bar{u} \right]. \label{2.18}
\end{eqnarray}
The $(1,2)$-component of the Lax equation (\ref{2.11}) with (\ref{2.16})--(\ref{2.18}) recovers the second-order equation (\ref{2.9}).
Two conserved quantities of the second-order equation (\ref{2.9}) follows 
from the alternative representation of $\det \Psi$ obtained from  (\ref{2.16})--(\ref{2.18}):
\begin{equation}\label{2.20}
\det \Psi = - \Psi_{11}^2 - \Psi_{12} \Psi_{21} = \frac{P(\lambda)}{(\lambda^2 - \lambda_1^2)^2(\lambda^2-\lambda_2^2)^2},
\end{equation}
where $P(\lambda)$ is the eight-degree polynomial given by
\begin{eqnarray}
P(\lambda) &=& \left[\lambda^4 - \lambda^2 \left(c+ \frac12\left|u\right|^2\right) + b \right]^2 \nonumber \\
&& + \lambda^2 \left[\lambda^2 u + \frac{i}{2} \frac{d u}{d x}-\frac{1}{2} u\left| u\right|^2 -c u \right]
\left[\lambda^2 \bar{u} - \frac{i}{2} \frac{d \bar{u}}{d x}-\frac{1}{2} \bar{u}\left| u\right|^2 -c \bar{u} \right].
\label{2.21}
\end{eqnarray}
It follows from (\ref{2.19}) that $\det \Psi$ is $x$ independent. Hence, 
the polynomial $P(\lambda)$ is $x$ independent. The coefficients of $P(\lambda)$ are $x$ independent if and only if 
solutions to the second-order equation (\ref{2.9}) also satisfy the following two first-order invariants
\begin{equation}\label{2.22}
2i \left( \bar{u} \frac{du}{dx} - u \frac{d \bar{u}}{dx} \right) - 3|u|^4- 4c|u|^2 = 4a,
\end{equation}
\begin{equation}\label{2.23}
2\left|\frac{d u}{dx}\right|^2 -|u|^6 -2c |u|^4 - 4(a+2b)|u|^2 = 8d,
\end{equation}
where $a$ and $d$ are two real-valued parameters in addition to real-valued parameters $b$ and $c$ of the second-order equation (\ref{2.9}). Substituting (\ref{2.22}) and (\ref{2.23}) into (\ref{2.21}) yields the polynomial 
$P(\lambda)$ in the form (\ref{polynomial}).

It follows from (\ref{2.19}) that $\det \Psi$  has only simple poles at $(\pm \lambda_1, \pm \lambda_2)$ if $\lambda_2 \neq \pm \lambda_1$ (see Remark \ref{remark-5}). Therefore, the two eigenvalues $\lambda_1$ and $\lambda_2$ are chosen from the eight roots of the polynomial $P(\lambda)$.

It remains to relate parameters $a$ and $d$ to $\lambda_1$, $\lambda_2$, $H$, and $F$. By equating (\ref{2.19}) and (\ref{2.20}) and substituting (\ref{polynomial}) for $P(\lambda)$, we derive
coefficients for even powers of $\lambda$. The coefficient of $\lambda^8$ is satisfied identically. The
coefficients of $\lambda^6$ and $\lambda^0$ recover relations (\ref{2.10}) for parameters $c$ and $b$ respectively,
and the coefficients of $\lambda^4$ and $\lambda^2$ yield the following relations for parameters $a$ and $d$
respectively, 
\end{proof}

The polynomial $P(\lambda)$ in (\ref{polynomial}) generally has four pairs of distinct roots, two of which must be chosen as the eigenvalues $\lambda_1$ and $\lambda_2$ of the algebraic method satisfying either the reduction $\lambda_2 = \bar{\lambda}_1$ with ${\rm Re}(\lambda_1) \neq 0$ or the reduction $\lambda_1, \lambda_2 \in i \mathbb{R}$ with $\lambda_2 \neq \pm \lambda_1$. We label the four pairs of distinct roots of $P(\lambda)$ as $\{ \pm \lambda_1, \pm\lambda_2, \pm\lambda_3, \pm \lambda_4\}$, where the complementary eigenvalues $\lambda_3$ and $\lambda_4$ are not used in the constraint (\ref{2.1}). The polynomial $P(\lambda)$ can be factorized by its roots as
\begin{equation}\label{3.22}
P(\lambda) = (\lambda^2  - \lambda_1^2) (\lambda^2-\lambda_2^2) (\lambda^2  - \lambda_3^2) (\lambda^2-\lambda_4^2).
\end{equation}
It follows by expanding (\ref{3.22}) in even powers of $\lambda$ and comparing it with (\ref{polynomial})  that 
\begin{equation}\label{3.23}
\left\{\begin{array}{l}
\lambda_1^2+\lambda_2^2+\lambda_3^2+ \lambda_4^2=2c,\\
(\lambda_1^2+\lambda_2^2)(\lambda_3^2+ \lambda_4^2)+\lambda_1^2 \lambda_2^2 +\lambda_3^2 \lambda_4^2=a+2b+c^2,\\
\lambda_1^2 \lambda_2^2 (\lambda_3^2+ \lambda_4^2) +
\lambda_3^2\lambda_4^2 (\lambda_1^2+\lambda_2^2)=ac+2bc-d,\\
\lambda_1^2 \lambda_2^2 \lambda_3^2 \lambda_4^2 =b^2.
\end{array}
\right.
\end{equation}
It follows from the second equation of 
(\ref{2.10}) and the first equation of (\ref{3.23}) that
\begin{equation}\label{3.24}
H=\frac12 (\lambda_3^2 + \lambda_4^2 -\lambda_1^2 - \lambda_2^2).
\end{equation}
It follows from the last equation of (\ref{3.23}) that two cases are possible for $b$: either $b = \lambda_1 \lambda_2 \lambda_3 \lambda_4$ or $b = -\lambda_1 \lambda_2 \lambda_3 \lambda_4$. The second choice, however, follows from the first one by replacing
$\lambda_4 \mapsto -\lambda_4$, so we will only consider the case
$b = \lambda_1 \lambda_2 \lambda_3 \lambda_4$. It follows from 
the first equation of (\ref{2.10}) with $b = \lambda_1 \lambda_2 \lambda_3 \lambda_4$ that 
\begin{equation}\label{3.24a}
F = \frac{\lambda_3\lambda_4}{\lambda_1 \lambda_2} - 1.
\end{equation}
Substituting (\ref{3.24}) and (\ref{3.24a}) into (\ref{2.10}) and (\ref{3.21}) allows us
to express parameters $a$, $b$, $c$, and $d$ in terms of the eigenvalues $\{ \lambda_1, \lambda_2, \lambda_3, \lambda_4\}$:
\begin{equation}\label{3.25}
\left\{ 
\begin{array}{lll}
\displaystyle 
a &=& -\frac{1}{4} [ (\lambda_1 + \lambda_2)^2 - (\lambda_3 + \lambda_4)^2]
[ (\lambda_1 - \lambda_2)^2 - (\lambda_3 - \lambda_4)^2],\\
\displaystyle 
b&=& \lambda_1 \lambda_2 \lambda_3 \lambda_4, \\
\displaystyle 
c&=& \frac12 (\lambda_1^2+\lambda_2^2+\lambda_3^2+\lambda_4^2),\\
\displaystyle 
d&=& -\frac 18 (\lambda_1^2+\lambda_2^2 -\lambda_3^2-\lambda_4^2)
(\lambda_1^2-\lambda_2^2+\lambda_3^2-\lambda_4^2)
(\lambda_1^2-\lambda_2^2-\lambda_3^2+\lambda_4^2).
\end{array}\right.
\end{equation}
We have checked that all equations of system (\ref{3.23}) are satisfied under the relations (\ref{3.25}).

Outcomes of the algebraic method are summarized as follows. 

\vspace{0.2cm}

\centerline{\fbox{\parbox[cs]{1,0\textwidth}{
The standing waves of the DNLS equation (\ref{dnls}) of the form (\ref{trav-wave}) satisfy the second-order equation (\ref{2.9}) and the first-order invariants (\ref{2.22}) and (\ref{2.23}) with four parameters $a$, $b$, $c$, and $d$. These parameters generally determine four distinct pairs of roots of the polynomial $P(\lambda)$ in (\ref{polynomial}) and (\ref{3.22}). The connection formulas (\ref{3.23}) are inverted in the form (\ref{3.25}). Picking any two distinct roots of the polynomial $P(\lambda)$ as two eigenvalues $\lambda_1$ and $\lambda_2$ of the algebraic method allows us to relate the standing wave of the form (\ref{trav-wave}) to squared eigenfunctions of the KN spectral problem (\ref{lax-1}) by either (\ref{constraint-red-2}) if $\lambda_2 = \bar{\lambda}_1$ with ${\rm Re}(\lambda_1) \neq 0$ or (\ref{constraint-red-1}) if $\lambda_1, \lambda_2 \in i \mathbb{R}$ with $\lambda_1 \neq \pm \lambda_2$. If the standing wave is $L$-periodic, 
so are the squared eigenfunctions due to relations (\ref{2.1}) and (\ref{2.8}). Then, the eigenvectors $\varphi = (p_1,q_1)^T$ and $\varphi = (p_2,q_2)^T$ for the eigenvalues $\lambda_1$ and $\lambda_2$ are either $L$-periodic or $L$-anti-periodic.}}}

\section{Modulational instability of periodic waves}
\label{sec:7}

Spectral stability of the standing waves of the form (\ref{trav-wave}) in the time evolution of the DNLS equation (\ref{dnls}) can be studied by adding
a perturbation $w$ of the form
\begin{equation}
\label{ansatz}
u(x,t) = e^{4ibt} \left[ \tilde{u}(x+2ct) + w(x+2ct,t) \right].
\end{equation}
Substituting (\ref{ansatz}) into (\ref{dnls}) and truncating at the linear terms in $w$ yields the linearized system of equations
\begin{equation}
\left\{ \begin{array}{l}
i w_t - 4 b w + 2 i c w_x + w_{xx} + i [2|u|^2 w_x + u^2 \bar{w}_x
+ 2(u \bar{u}_x + u_x \bar{u}) w + 2 u u_x \bar{w}] = 0, \\
-i \bar{w}_t - 4 b \bar{w} - 2 i c \bar{w}_x + \bar{w}_{xx}
- i [2|u|^2 \bar{w}_x + \bar{u}^2 w_x
+ 2(u \bar{u}_x + u_x \bar{u}) \bar{w} + 2 \bar{u} \bar{u}_x w] = 0,
\end{array} \right.
\label{dnls-lin}
\end{equation}
where the tilde notation for $u$ has been dropped as before. Variables can be separated in the linearized system (\ref{dnls-lin}) by 
\begin{equation}
\label{w-eigen}
w(x,t) = w_1(x) e^{t \Lambda}, \quad \bar{w}(x,t) = w_2(x) e^{t \Lambda},
\end{equation}
where $w_1$, $w_2$, and $\Lambda$ are found from the spectral problem 
\begin{eqnarray}
\left[ \begin{array}{cc} 
\mathcal{L}
& - i u^2 \partial_x - 2i u u_x \\
i \bar{u}^2 \partial_x + 2i \bar{u} \bar{u}_x &
\bar{\mathcal{L}}
\end{array} \right] \;
\left[ \begin{array}{c} w_1 \\ w_2 \end{array} \right] 
= i \Lambda \sigma_3 \left[ \begin{array}{c} w_1 \\ w_2 \end{array} \right],
\label{dnls-spectral}
\end{eqnarray}
where $\mathcal{L}:= 4 b - 2 i c \partial_x - \partial_x^2 
- 2 i |u|^2 \partial_x - 2i (u \bar{u}_x + u_x \bar{u})$ 
and $\sigma_3 = {\rm diag}(1,-1)$. Note that $w_1 \neq \bar w_2$ if $\Lambda \notin \mathbb{R}$. 

Our goal is to find the admissible values of $\Lambda$ for which $w_1$ and $w_2$ are bounded functions of $x$ on $\mathbb{R}$. By Floquet's theorem \cite{Kuchment}, the admissible values of $\Lambda$ form continuous spectral bands on the complex $\Lambda$-plane. The union of all admissible values 
of $\Lambda$ is referred to as the Floquet spectrum for the 
spectral problem (\ref{dnls-spectral}). 

The spectral and modulational instability of the standing wave are defined as follows.

\begin{definition}
	\label{def-stab}
If there exists $\Lambda$ with ${\rm Re}(\Lambda) > 0$ for which 
$(w_1,w_2) \in L^{\infty}(\mathbb{R})$ in (\ref{dnls-spectral}), then 
the standing wave of the form (\ref{trav-wave}) is called spectrally unstable. It is called modulationally unstable if the unstable spectral band with ${\rm Re}(\Lambda) > 0$ intersects the origin in the $\Lambda$-plane.
\end{definition}

\begin{remark}
	Modulational instability in Definition \ref{def-stab} is a subset of spectral instability. It occurs if the perturbations of increasingly long spatial periods grow in time \cite{BHJ}. The importance of the modulational instability is illustrated in the formation of rogue waves on the modulationally unstable background. The rogue waves are not fully localized in space and time if the periodic standing wave background is modulationally stable  \cite{CPW,PW}.
\end{remark}

There exists an explicit relation between the admissible values of $\Lambda$ 
for which $(w_1,w_2) \in L^{\infty}(\mathbb{R})$ 
and suitable solutions of the Lax equations
(\ref{lax-1-intro}) and (\ref{lax-2-intro}).
By substituting the standing waves of the form (\ref{trav-wave}) into 
the Lax equations (\ref{lax-1-intro}) and (\ref{lax-2-intro}) and separating the variables in the form
\begin{equation}
\label{varphi-separation}
\varphi(x,t) = e^{2ibt \sigma_3} \tilde{\varphi}(x+2ct,t), 
\end{equation}
we obtain the following system of linear equations:
\begin{equation}\label{lax-1-stab}
\varphi_x = U\varphi, \quad
\varphi_t + 2 i b \sigma_3 \varphi + 2 c \varphi_x =V\varphi, 
\end{equation}
where
\begin{equation}
\label{potentials}
 U=\left(\begin{array}{cc}
-i \lambda^2 & \lambda u\\
-\lambda \bar{u} & i \lambda^2\\
\end{array}
\right), \;\;
V = \left(\begin{array}{cc}
-2 i \lambda^4 + i \lambda^2 |u|^2 & 2 \lambda^3 u + \lambda (i u_x - |u|^2 u)\\
-2 \lambda^3 \bar{u} + \lambda (i \bar{u}_x + |u|^2 \bar{u}) & 2 i \lambda^4 - i \lambda^2 |u|^2\\
\end{array}
\right),
\end{equation}
and the tilde notations for $\varphi$ and $u$ have been dropped again. 
We note that $U$ and $V$ in (\ref{potentials}) are $t$-independent since the transformed solution $u$ (former $\tilde{u}$) is a function of $x$ only.
The following proposition summarizes the result obtained in \cite{ChenYang}. 

\begin{proposition}
	\label{prop-relation}
Let $\varphi = (\varphi_1,\varphi_2)^T$ be the eigenvector of the Lax system (\ref{lax-1-stab}) for the eigenvalue $\lambda \in \mathbb{C}$. Then the perturbation $w$ satisfying the linearized DNLS equation (\ref{dnls-lin}) is expressed by 
\begin{equation}
\label{squared-eigen}
w = \partial_x \varphi_1^2, \quad \bar{w} = \partial_x \varphi_2^2.
\end{equation}
Consequently, if $\varphi(t,x) = \chi(x) e^{t \Omega}$, then 
$w_1 = \partial_x \chi_1^2$, $w_2 = \partial_x \chi_2^2$, and
$\Lambda = 2\Omega$ in (\ref{w-eigen}). Moreover, $\Omega$ is related 
to $\lambda$ by 
\begin{equation}
\label{lambda-stab}
\Omega = \pm 2 i \sqrt{P(\lambda)},
\end{equation}
where $P(\lambda)$ is given by (\ref{polynomial}).
\end{proposition}

\begin{proof}
By the linear superposition principle, it suffices to show that
\begin{equation}
\label{squared-eigen-new}
w = - i \lambda \varphi_1^2 + u \varphi_1 \varphi_2, \quad
\bar{w} = i \lambda \varphi_2^2 - \bar{u} \varphi_1 \varphi_2,
\end{equation}
satisfies the linearized DNLS equation (\ref{dnls-lin}) if $\varphi = (\varphi_1,\varphi_2)^T$ satisfies the Lax equations 
(\ref{lax-1-stab}). The two terms
in (\ref{squared-eigen-new}) are inspected separately as follows:
\begin{eqnarray*}
	&& i \partial_t(\varphi_1^2) - 4 b \varphi_1^2 + 2 i c \partial_x (\varphi_1^2) + \partial_x^2 (\varphi_1^2) \\
	&& +
	i [2|u|^2 \partial_x  (\varphi_1^2) - u^2 \partial_x  (\varphi_2^2)
	+ 2(u \bar{u}_x + u_x \bar{u}) \varphi_1^2 - 2 u u_x \varphi_2^2] \\
	&& = 4 u^2 \varphi_2^2 \lambda^2 + 4 i |u|^2 u \varphi_1 \varphi_2 \lambda
	+ 2 i (u \bar{u}_x + u_x \bar{u}) \varphi_1^2 - 2 i u u_x \varphi_2^2,
\end{eqnarray*}
and
\begin{eqnarray*}
	&& i \partial_t (u \varphi_1 \varphi_2) - 4 b u \varphi_1 \varphi_2 + 2 i c \partial_x (u \varphi_1 \varphi_2) + \partial_x^2 (u \varphi_1 \varphi_2) \\
	&& +
	i [2|u|^2 \partial_x  (u \varphi_1 \varphi_2) - u^2 \partial_x  (\bar{u} \varphi_1 \varphi_2)
	+ 2(u \bar{u}_x + u_x \bar{u}) u \varphi_1 \varphi_2 - 2 u u_x \bar{u} \varphi_1 \varphi_2] \\
	&& = 4 i u^2 \varphi_2^2 \lambda^3 - 4 |u|^2 u \varphi_1 \varphi_2 \lambda^2
	- 2 (u \bar{u}_x + u_x \bar{u}) \varphi_1^2 \lambda + 2 u u_x \varphi_2^2 \lambda.
\end{eqnarray*}
Summing the first equality multiplied by $(-i \lambda)$ and the second equality yields zero which verifies the relations (\ref{squared-eigen-new}).

In order to show (\ref{lambda-stab}), we recall that $U$ and $V$ in (\ref{potentials}) are $t$-independent. Hence we can separate
the variables in the form $\varphi(t,x) = \chi(x) e^{t \Omega}$
and obtain $\Omega$ from the characteristic equation
\begin{equation}
\label{det-eq}
{\rm det}(\Omega + 2 i b \sigma_3 + 2 c U - V) = 0.
\end{equation}
By expanding the determinant and using first-order invariants (\ref{2.22}) and (\ref{2.23}), we verify that the characteristic equation (\ref{det-eq}) is equivalent to 
\begin{equation}
\label{Omega-eq}
\Omega^2 + 4 P(\lambda) = 0, 
\end{equation}
which yields (\ref{lambda-stab}) after extracting the square root.
\end{proof}

\begin{remark}
Roots of the polynomial $P(\lambda)$ in (\ref{polynomial}) 
are mapped to the origin of the $\Lambda$ plane.
\end{remark}

The following two propositions state explicitly the stability results on  the DNLS equation which follow from Theorem 9 in \cite{DU2} (see also their Section 6.1). 

\begin{proposition}
	\label{prop-stability-real}
	Assume that $P(\lambda)$ is given by (\ref{3.22}) with the roots
	$(\pm \lambda_1,\pm \lambda_2,\pm \lambda_3,\pm \lambda_4) \in \mathbb{C}\backslash \mathbb{R}$. If $\lambda \in \mathbb{R}$, then $\Lambda \in i \mathbb{R}$.
\end{proposition}

\begin{proof}
It follows from (\ref{3.22}) that if $(\pm \lambda_1,\pm \lambda_2,\pm \lambda_3,\pm \lambda_4) \in \mathbb{C}\backslash \mathbb{R}$, then 
	$P(\lambda) > 0$ for every $\lambda \in \mathbb{R}$. 
Indeed, $P(\lambda) \sim \lambda^8$ as $|\lambda| \to \infty$ 
and $P(\lambda)$ has no real roots. If $P(\lambda) > 0$, 
then $\Omega \in i \mathbb{R}$ in (\ref{lambda-stab}) so that 
$\Lambda = 2 \Omega \in i \mathbb{R}$.
\end{proof}

\begin{proposition}
	\label{prop-stability}
	Assume that $P(\lambda)$ is given by (\ref{3.22}) with the roots
$(\pm \lambda_1,\pm \lambda_2,\pm \lambda_3,\pm \lambda_4)$. If 
$\lambda \in i \mathbb{R}$, then $\Lambda \in i \mathbb{R}$, provided the following conditions are true: 
\begin{itemize}
	\item The roots form two complex quadruplets or four pairs of double real eigenvalues.
	\item Two pairs of roots are purely imaginary, e.g. $\lambda_{3,4} = i \beta_{3,4}$ with $0 < \beta_4 < \beta_3$, and 
\begin{equation}
\label{stab-1}
{\rm Im}(\lambda) \in (-\infty, -\beta_3] \cup  [-\beta_4,\beta_4] \cup [\beta_3,\infty).
\end{equation}
\item Four pairs of roots are purely imaginary, e.g. $\lambda_{1,2,3,4} = i \beta_{1,2,3,4}$ with $0 < \beta_4 < \beta_3 < \beta_2 < \beta_1$, and
\begin{equation}
\label{stab-2}
{\rm Im}(\lambda) \in (-\infty, -\beta_1] \cup [-\beta_2,-\beta_3] \cup [-\beta_4,\beta_4] \cup [\beta_3,\beta_2] \cup [\beta_1,\infty).
\end{equation}
\end{itemize}
\end{proposition}

\begin{proof}
For $\lambda \in i \mathbb{R}$, we can rewrite the polynomial 
$P(\lambda)$ given by (\ref{3.22}) in the form
\begin{equation}\label{3.22a}
	P(z) = (z + \lambda_1^2) (z + \lambda_2^2) (z + \lambda_3^2) (z + \lambda_4^2),
\end{equation}
where $z = -\lambda^2 \geq 0$ and the notation for $P(z)$ has been overwritten. 

If the roots form complex quadruplets or double real eigenvalues, then $P(z) > 0$ for every $z \geq 0$. Indeed, $P(z) \sim z^4$ as $|z| \to \infty$ 
	and $P(z)$ has no real roots on $\mathbb{R}_+$. If $P(z) > 0$, then $\Omega \in i \mathbb{R}$ in (\ref{lambda-stab}) so that $\Lambda = 2 \Omega \in i \mathbb{R}$.

If two or four pairs of purely imaginary eigenvalues occur, then $P(z) > 0$ for either $z \in (0,\beta_4^2) \cup (\beta_3^2,\infty)$ or $z \in (0,\beta_4^2) \cup (\beta_3^2,\beta_2^2) \cup (\beta_1^2,\infty)$, respectively. This gives the respective constraints (\ref{stab-1}) and (\ref{stab-2}) on the admissible values of $\lambda \in i \mathbb{R}$, for which $\Lambda \in i \mathbb{R}$.
\end{proof}

\begin{remark}
It was also proven in Theorem 9 in \cite{DU2} that if 
$\Lambda \in i \mathbb{R}$ for a given $\lambda \in \mathbb{R} \cup i \mathbb{R}$, then $\lambda \in \mathbb{R} \cup i \mathbb{R}$ belongs to the Floquet spectrum of the KN spectral problem (\ref{lax-1-intro}). By Propositions \ref{prop-stability-real} and \ref{prop-stability}, this implies that $\mathbb{R} \cup i \mathbb{R} \backslash S$ belongs to the Floquet spectrum of the KN spectral problem (\ref{lax-1-intro}), where $S \subset i \mathbb{R}$ includes either two or four spectral gaps in (\ref{stab-1}) and (\ref{stab-2}), respectively. Our numerical results suggest that $\lambda \in S \subset i \mathbb{R}$ is not in the Floquet spectrum of the KN spectral problem (\ref{lax-1-intro}).
\end{remark}

Outcomes of the modulational stability analysis are summarized as follows.

\vspace{0.2cm}

\centerline{\fbox{\parbox[cs]{1,0\textwidth}{If we compute the admissible values of $\lambda$ in the Floquet spectrum of the Lax system (\ref{lax-1-stab}) for the standing waves of the form 
(\ref{trav-wave}), then we can obtain the admissible values of 
$\Lambda$ in the Floquet spectrum of the stability problem (\ref{dnls-spectral})
by using $\Lambda = 2 \Omega = \pm 4i \sqrt{P(\lambda)}$. 
By Propositions \ref{prop-stability-real} and \ref{prop-stability}, 
spectral instability of the standing waves may only arise if there are admissible values of $\lambda$ in open quadrants of the complex plane or on the imaginary axis in either two or four spectral gaps (\ref{stab-1}) and (\ref{stab-2}) respectively.
}}}

\section{Classification of periodic standing waves}
\label{sec:3}

Here we characterize the periodic standing waves of the DNLS equation 
(\ref{dnls}) by solving the second-order equation (\ref{2.9}) closed with the first-order invariants (\ref{2.22}) and (\ref{2.23}).

We use the polar form $u(x) = R(x) e^{i\Theta(x)}$ with real-valued $R(x)$ and $\Theta(x)$ for the periodic standing waves. Substituting the polar form into the first-order invariants (\ref{2.22}) and (\ref{2.23}) leads to
\begin{equation}
\label{3.5}
4R^2 \frac{d \Theta}{d x} +3R^4 +4cR^2=-4a \quad \Rightarrow \quad
\frac{d \Theta}{d x} =-\frac{a}{R^2}  -\frac 34 R^2 - c,
\end{equation}
and
\begin{equation}\label{3.1}
2\left(\frac{d R}{d x}\right)^2+ 2R^2 \left(\frac{d \Theta}{d x}\right)^2 -R^6 -2cR^4 -4(a+2b) R^2 =8d.
\end{equation}
Inserting (\ref{3.5}) into (\ref{3.1}) yields the first-order quadrature:
\begin{equation}\label{3.6}
\left(\frac{d R}{d x}\right)^2+ \frac{a^2}{R^2} + \frac{1}{16} R^6 + \frac{c}{2} R^4
+ R^2 \left(c^2 - 4b -\frac{a}{2}\right) + 2ac - 4d = 0.
\end{equation}

Two cases are distinguished here: $a \neq 0$ and $a = 0$. 
In the remainder of this section, we will consider the general case $a \neq 0$ and obtain the periodic solutions in an explicit form. In the following section, 
we will set $a = 0$ and investigate the periodic solutions and their modulational instability in more details. 

If $a \neq 0$, the singularity $R=0$ of the quadrature (\ref{3.6}) is unfolded with the transformation
$\rho = \frac{1}{2} R^2$ which yields
\begin{equation}\label{3.8}
\left(\frac{d \rho}{d x}\right)^2 + Q(\rho) = 0,
\end{equation}
where $Q(\rho)$ is the quartic polynomial given by
\begin{equation}
\label{quartic-Q}
Q(\rho) = \rho^4 +4 c \rho^3 +2(2c^2 -a -8b)\rho^2 +4(ac-2d)\rho +a^2.
\end{equation}
The polynomial $Q(\rho)$ can be factorized
by its roots $(u_1,u_2,u_3,u_4)$ as
\begin{eqnarray}
Q(\rho) = (\rho-u_1)(\rho-u_2)(\rho-u_3)(\rho-u_4). 
\label{3.9}
\end{eqnarray}
Equating coefficients of the same powers in (\ref{quartic-Q}) and (\ref{3.9}) yields
\begin{equation}\label{3.10}\left\{\begin{array}{l}
u_1+u_2+u_3+u_4=- 4c,\\
u_1u_2+u_1u_3+u_1u_4+u_2u_3+u_2u_4+u_3u_4=2(2c^2 -a -8b),\\
u_1u_2u_3+u_1u_2u_4+u_1u_3u_4+u_2u_3u_4=4(2d- ac),\\
u_1u_2u_3u_4=a^2.\\
\end{array}
\right.
\end{equation}

Recall that the parameters $a$, $b$, $c$, and $d$ are related to the 
roots $(\pm \lambda_1,\pm \lambda_2,\pm \lambda_3,\pm \lambda_4)$ of the polynomial $P(\lambda)$ by the transformation formulas (\ref{3.23}) and (\ref{3.25}). The following proposition shows that the roots of $P(\lambda)$ are related to the roots of $Q(\rho)$ by using simple and explicit expressions. The same relations were found before (see Eqs. (3.17) in \cite{Kam1}) by using the so-called resolvent method. 

\begin{proposition}
	\label{prop-roots}
Let $(\pm \lambda_1,\pm \lambda_2,\pm \lambda_3,\pm \lambda_4)$ be the roots of $P(\lambda)$ in (\ref{3.22}) and $(u_1,u_2,u_3,u_4)$ be roots of $Q(\rho)$ in (\ref{3.9}). Then
\begin{equation}\label{3.26}
\left\{\begin{array}{lll}
u_1&=&-\frac12 (\lambda_1-\lambda_2+\lambda_3-\lambda_4)^2,\\
u_2&=&-\frac12 (\lambda_1-\lambda_2-\lambda_3+\lambda_4)^2,\\
u_3&=&-\frac12 (\lambda_1+\lambda_2-\lambda_3-\lambda_4)^2,\\
u_4&=&-\frac12 (\lambda_1+\lambda_2+\lambda_3+\lambda_4)^2.
\end{array}\right.
\end{equation}
\end{proposition}

\begin{proof}
We substitute the roots $(u_1,u_2,u_3,u_4)$ expressed by (\ref{3.26}) into system (\ref{3.10}). The last equation of system (\ref{3.10}) yields the first equation of system (\ref{3.25}) after extracting the negative square root. The first equation of (\ref{3.10}) yields the third equation of (\ref{3.25}),
$$
u_1 + u_2 + u_3 + u_4 = -2(\lambda_1^2 + \lambda_2^2 + \lambda_3^2 + \lambda_4^2) = -4c.
$$
Similarly, the second equation of (\ref{3.10}) is compatible 
with system (\ref{3.25}) due to 
\begin{eqnarray*}
&&	(u_1+u_2)(u_3+u_4) + u_1 u_2 + u_3 u_4 \\
&& =  (\lambda_1^2 -\lambda_2^2)^2 + (\lambda_3^2 - \lambda_4^2)^2 + 2(\lambda_1^2 + \lambda_2^2) (\lambda_3^2 +\lambda_4^2) - 8 \lambda_1 \lambda_2 \lambda_3 \lambda_4 \\
&& \quad + \frac{1}{2} (\lambda_1^4 + 6 \lambda_1^2 \lambda_2^2 + \lambda_2^4) + \frac{1}{2} (\lambda_3^4 + 6 \lambda_3^2 \lambda_4^2 + \lambda_4^4) -(\lambda_1^2 + \lambda_2^2)(\lambda_3^2 + \lambda_4^2) - 4 \lambda_1 \lambda_2 \lambda_3 \lambda_4 \\
&& = (\lambda_1^2 + \lambda_2^2 + \lambda_3^2 + \lambda_4^2)^2  - 16 \lambda_1 \lambda_2 \lambda_3 \lambda_4 + \frac{1}{2}  [ (\lambda_1 + \lambda_2)^2 - (\lambda_3 + \lambda_4)^2] [ (\lambda_1 - \lambda_2)^2 - (\lambda_3 - \lambda_4)^2] \\
&& = 4c^2 - 2 a - 16b.
\end{eqnarray*}
Compatibility of the third equation of (\ref{3.10}) is checked with Wolfram's Mathematica.	
\end{proof}

Because the coefficients of $Q$ are real-valued, we have three cases to consider:
(i) four roots of $Q$ are real, (ii) two roots of $Q$ are real and one pair of roots is complex-conjugate, and (iii) two pairs of roots of $Q$ are complex-conjugate. Each case is considered separately.

\subsection{Four roots of $Q$ are real}

For simplicity, we order the four real roots of $Q$ as
\begin{equation}
\label{order-roots}
u_4\leq u_3 \leq u_2 \leq u_1.
\end{equation}
Periodic solutions to the quadrature (\ref{3.8}) with (\ref{3.9}) and (\ref{order-roots}) can be expressed explicitly (see, e.g., \cite{CPgardner}) by
\begin{equation}\label{3.14}
\rho(x) = u_4 +
\frac{(u_1-u_4)(u_2-u_4)}{(u_2-u_4)+(u_1-u_2){\rm sn}^2 (\nu x;k)},
\end{equation}
where positive parameters $\nu$ and $k$ are uniquely expressed by
\begin{equation}\label{3.13}
\nu = \frac{1}{2} \sqrt{(u_1-u_3)(u_2-u_4)}, \quad
k = \frac{\sqrt{(u_1-u_2)(u_3-u_4)}}{\sqrt{(u_1-u_3)(u_2-u_4)}}.
\end{equation}
The periodic solution $\rho$ in (\ref{3.14}) is located in the interval $[u_2,u_1]$ and has period
$L = 2 K(k) \nu^{-1}$. The solution is meaningful for 
$\rho = \frac{1}{2} R^2 \geq 0$ if and
only if $u_2 \geq 0$. The four pairs of eigenvalues $\{ \pm \lambda_1, \pm \lambda_2, \pm \lambda_3, \pm \lambda_4 \}$ generate real roots
$\{u_1,u_2,u_3,u_4\}$ of $Q$ by the transformation formula (\ref{3.26}) 
if and only if they satisfy the following three configurations:
(i) they form two complex quadruplets; (ii) they form four pairs of purely imaginary 
eigenvalues; or (iii) they form four pairs of real eigenvalues. Each case is considered separately.

\subsubsection{Two complex quadruplets} Assume that 
the four pairs of eigenvalues
$\{ \pm \lambda_1, \pm \lambda_2, \pm \lambda_3, \pm \lambda_4 \}$
form two complex quadruplets with
\begin{equation}
\label{configuration-1}
\lambda_1 = \bar{\lambda}_2 = \alpha_1 + i \beta_1, \qquad
\lambda_3 = \bar{\lambda}_4 = \alpha_2 + i \beta_2.
\end{equation}
Then the roots ordered as (\ref{order-roots}) satisfy the more precise ordering
\begin{equation}
\label{order-roots-precise}
u_4 \leq u_3 \leq 0 \leq u_2 \leq u_1.
\end{equation}
If $\alpha_1,\alpha_2,\beta_1,\beta_2$ are all positive, so that $\lambda_1$ and $\lambda_3$ are located in the first quadrant, we deduce the explicit expressions
\begin{equation}
\label{3.27}
\left\{ \begin{array}{l}
\alpha_1 = \frac{1}{2\sqrt{2}}(\sqrt{-u_4} + \sqrt{-u_3}), \\
\alpha_2 = \frac{1}{2\sqrt{2}}(\sqrt{-u_4} - \sqrt{-u_3}), \end{array}
\right.
\qquad
\left\{ \begin{array}{l}
\beta_1 = \frac{1}{2\sqrt{2}}(\sqrt{u_1} + \sqrt{u_2}), \\
\beta_2 = \frac{1}{2\sqrt{2}}(\sqrt{u_1} - \sqrt{u_2}), \end{array}
\right.
\end{equation}
so that $\alpha_2 \leq \alpha_1$ and $\beta_2 \leq \beta_1$.

\subsubsection{Four pairs of purely imaginary eigenvalues}
Assume that the four pairs of eigenvalues
$\{ \pm \lambda_1, \pm \lambda_2, \pm \lambda_3, \pm \lambda_4 \}$
are purely imaginary with
\begin{equation}
\label{configuration-2}
\lambda_1 = i \beta_1, \qquad
\lambda_2 = i \beta_2, \qquad
\lambda_3 = i \beta_3, \qquad
\lambda_4 = i \beta_4.
\end{equation}
Then the roots ordered as (\ref{order-roots}) satisfy the more precise ordering
\begin{equation}
\label{order-roots-precise-1}
0 \leq u_4 \leq u_3  \leq u_2 \leq u_1.
\end{equation}
It follows from (\ref{3.26}) that
\begin{equation}
\label{3.27a}
\left\{ \begin{array}{l}
\beta_1 = \frac{1}{2\sqrt{2}}(\sqrt{u_1} + \sqrt{u_2} + \sqrt{u_3} + \sqrt{u_4}), \\
\beta_2 = \frac{1}{2\sqrt{2}}(- \sqrt{u_1} - \sqrt{u_2}
+ \sqrt{u_3} + \sqrt{u_4}), \\
\beta_3 = \frac{1}{2\sqrt{2}}(\sqrt{u_1} - \sqrt{u_2} - \sqrt{u_3}
+ \sqrt{u_4}), \\
\beta_4 = \frac{1}{2\sqrt{2}}(- \sqrt{u_1} + \sqrt{u_2} - \sqrt{u_3}
+ \sqrt{u_4}),\end{array}
\right.
\end{equation}
so that $\beta_2 \leq \beta_4 \leq \beta_3 \leq \beta_1$. 

\begin{remark}
	In the case of the ordering (\ref{order-roots-precise-1}), 
	exchanging $u_1$ with $u_3$ and $u_2$ with $u_4$ 
	generates another periodic solution in the form
	\begin{equation}\label{3.14a}
	\rho(x) = u_2 -
	\frac{(u_2-u_3)(u_2-u_4)}{(u_2-u_4)-(u_3-u_4){\rm sn}^2 (\nu x;k)},
	\end{equation}
	with the same values of parameters $\nu$ and $k$ in (\ref{3.13}). 
	The periodic solution $\rho$ in (\ref{3.14a}) is located in the interval 
	$[u_4,u_3]$ and has the same period $L = 2 K(k) \nu^{-1}$.	
\end{remark}

\begin{remark}
	The ordering of $\beta_{1,2,3,4}$ in (\ref{3.27a}) corresponds to the transformation (\ref{3.26}) but is different from the ordering used in Proposition \ref{prop-stability} for the spectral gaps in (\ref{stab-2}).
\end{remark}

\subsubsection{Four pairs of real eigenvalues}

Assume that all pairs of eigenvalues
$\{ \pm \lambda_1, \pm \lambda_2, \pm \lambda_3, \pm \lambda_4 \}$
are real. Then the roots satisfy
\begin{equation}
\label{order-roots-precise-2}
u_4 \leq u_3 \leq u_2 \leq u_1 \leq 0.
\end{equation}
The solution (\ref{3.14}) is not meaningful because 
$\rho = \frac{1}{2} R^2 \leq 0$.

\subsection{Case: two roots of $Q$ are real and one pair of roots is complex-conjugate}

Let $u_{1,2}$ be real roots ordered as $u_2\leq u_1$ and
$u_{3,4} = \gamma\pm i\eta$ be complex-conjugate  roots with
\begin{equation}
\label{order-roots-complex}
u_2 \leq u_1, \quad u_3 = \gamma + i \eta, \quad u_4= \gamma - i \eta.
\end{equation}
Periodic solutions
to the quadrature (\ref{3.8}) with (\ref{3.9}) and (\ref{order-roots-complex})
can be expressed explicitly (see, e.g., \cite{CPgardner}) by
\begin{equation}\label{3.20}
\rho(x)= u_1 + \frac{(u_2-u_1)(1-{\rm cn} (\mu x;k))}{1+\delta +(\delta-1){\rm cn} (\mu x;k)},
\end{equation}
where positive parameters $\delta$, $\mu$, and $k$ are uniquely expressed by
\begin{equation}\label{3.18}
\delta = \frac{\sqrt{(u_2-\gamma)^2 + \eta^2}}{\sqrt{(u_1-\gamma)^2 + \eta^2}},\qquad
\mu = \sqrt[4]{\left[(u_1 - \gamma)^2+\eta^2 \right]\left[(u_2- \gamma)^2+\eta^2 \right]},
\end{equation}
and
\begin{equation}
2 k^2 = 1- \frac{(u_1 - \gamma)(u_2 - \gamma)+ \eta^2}{\sqrt{\left[(u_1 - \gamma)^2+\eta^2 \right]\left[(u_2- \gamma)^2+\eta^2 \right]}}.
\end{equation}
The periodic solution $\rho$ in (\ref{3.20}) is located in the interval $[u_2,u_1]$ and has period $L = 4 K(k) \mu^{-1}$. The solution is meaningful for $\rho = \frac{1}{2} R^2 \geq 0$ if and only if $u_2 \geq 0$. The four pairs of eigenvalues $\{ \pm \lambda_1, \pm \lambda_2, \pm \lambda_3, \pm \lambda_4 \}$ generate the two roots and one pair of complex-conjugate roots $\{u_1,u_2,u_3,u_4\}$ of $Q$ by the transformation formula (\ref{3.26}) 
if and only if they satisfy the following two configurations:
(i) they form one complex quadruplet and two pairs of purely imaginary eigenvalues; or (ii) they form one complex quadruplet and two pairs of real eigenvalues. Each case is considered separately.

\subsubsection{One complex quadruple and two pairs of purely imaginary eigenvalues} Assume that the four pairs of eigenvalues
$\{ \pm \lambda_1, \pm \lambda_2, \pm \lambda_3, \pm \lambda_4 \}$
form one quadruplet $\{ \pm \lambda_1, \pm \bar{\lambda}_1\}$ of complex eigenvalues
and two pairs $\{ \pm \lambda_3, \pm \lambda_4 \}$ of purely imaginary eigenvalues. Then, we have $0 \leq u_2 \leq u_1$ and $u_3 = \bar{u}_4$.
By writing
\begin{equation}
\label{configuration-3}
\lambda_1 = \bar{\lambda}_2 = \alpha_1 + i \beta_1, \qquad
\lambda_3 = i \beta_3, \qquad
\lambda_4 = i \beta_4.
\end{equation}
with positive $\alpha_1$ and $\beta_1$, we deduce the explicit expressions
\begin{equation}
\label{a3.16}
\left\{ \begin{array}{l}
\alpha_1 = \frac{1}{2} \sqrt{\sqrt{\gamma^2 + \eta^2}-\gamma}, \\
\beta_1 = \frac{1}{2\sqrt{2}}(\sqrt{u_1} + \sqrt{u_2}), \end{array}
\right.
\qquad
\left\{ \begin{array}{l}
\beta_3 = \frac{\eta}{2\sqrt{\sqrt{\gamma^2 + \eta^2}-\gamma}} + \frac{1}{2\sqrt{2}}(\sqrt{u_1} - \sqrt{u_2}), \\
\beta_4 = \frac{\eta}{2\sqrt{\sqrt{\gamma^2 + \eta^2}-\gamma}} - \frac{1}{2\sqrt{2}}(\sqrt{u_1} - \sqrt{u_2}),
\end{array}
\right.
\end{equation}
so that $\beta_4 \leq \beta_3$. 

\subsubsection{One complex quadruple and two pairs of real eigenvalues}

Assume that the four pairs of eigenvalues form one quadruplet $\{ \pm \lambda_1, \pm \bar{\lambda}_1\}$ of complex eigenvalues
and two pairs $\{ \pm \lambda_3, \pm \lambda_4 \}$ of real eigenvalues. Then $u_1 = \bar{u}_2$ and $u_4 \leq u_3 \leq 0$. The solution (\ref{3.14}) is not meaningful because $\rho = \frac{1}{2} R^2 \leq 0$.

\subsection{Case: two pairs of roots of $Q$ are complex-conjugate}

In the case of no real roots of $Q$, we have $Q(\rho) > 0$ for every $\rho \in \mathbb{R}$. There exist no periodic wave solutions to the quadrature (\ref{3.8}) with $Q(\rho) > 0$ 
in the space of real functions for $\rho$. Hence, this case
does not result in the periodic wave solutions.

\vspace{0.2cm}

Outcomes of the classification of the periodic standing waves 
in the DNLS equation (\ref{dnls}) are summarized as follows.

\vspace{0.2cm}

\centerline{\fbox{\parbox[cs]{1,0\textwidth}{
There exists exactly two families of periodic standing waves expressed by either (\ref{3.14}) or (\ref{3.20}) for $\rho = \frac{1}{2} R^2$. The family (\ref{3.14}) is related to either two complex quadruplets in the case 
(\ref{order-roots-precise}) or four pairs of purely imaginary 
eigenvalues in the case (\ref{order-roots-precise-1}). 
The family (\ref{3.20}) is related to one complex quadruplet and 
two pairs of purely imaginary eigenvalues in the case (\ref{order-roots-complex}).}}}

\section{Periodic standing waves in the case of $a = 0$}
\label{sec:4}

The family of periodic standing waves $u(x) = R(x) e^{i \Theta(x)}$ can be made explicit
in the case $a = 0$. This case for the NLS equation is referred to as the waves of
trivial phase in \cite{DS} (see also \cite{CPW}). For the DNLS equation,
the phase is still nontrivial for $a = 0$ due to the dependence of $\Theta$ from $R^2$ in (\ref{3.5}). The case of $a = 0$ was the only case of periodic standing wave solutions of the DNLS equation considered in \cite{Hakaev}.

It follows from (\ref{3.6}) with $a = 0$ that the amplitude function $R$ satisfies the quadrature
\begin{equation}\label{R-quadrature}
\left(\frac{d R}{d x}\right)^2 + F(R) = 4d,
\end{equation}
where
\begin{equation}
\label{R-roots}
F(R) = \frac{1}{16} R^6 + \frac{c}{2} R^4
+ (c^2 - 4b) R^2.
\end{equation}
There is no singularity at $R = 0$ if $a = 0$.

\begin{figure}[hb!]
	\includegraphics[width=0.48\textwidth]{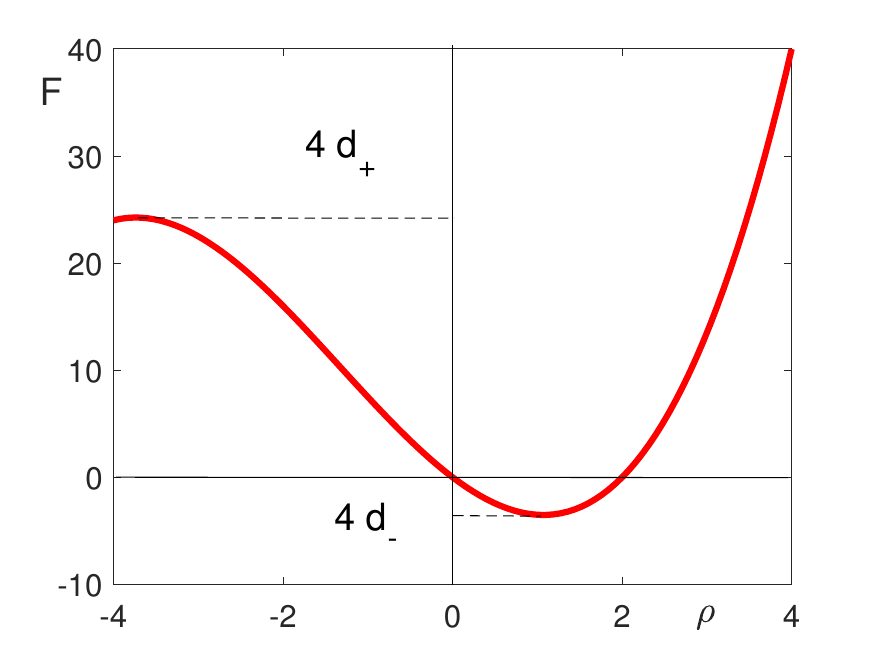}
	\includegraphics[width=0.48\textwidth]{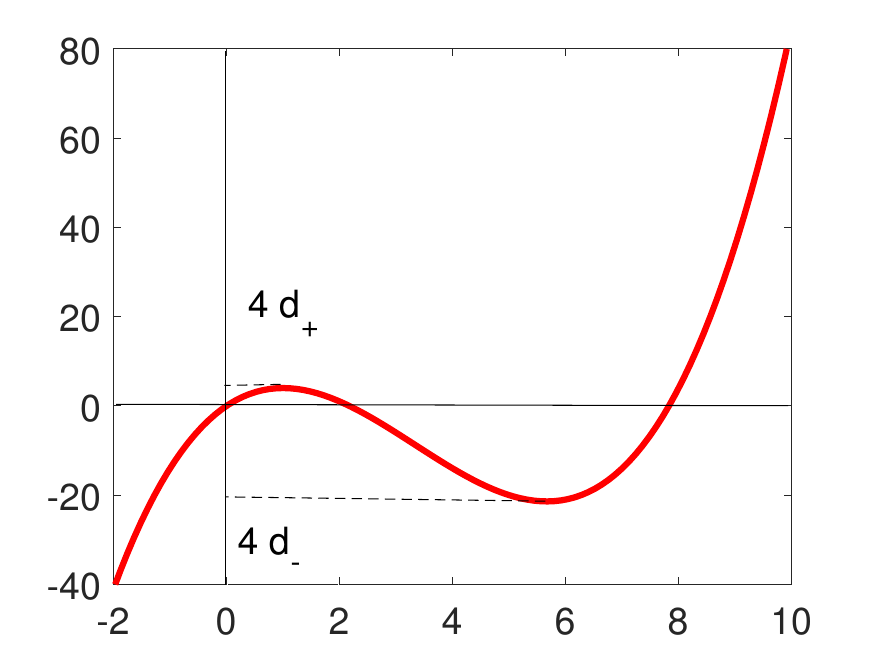}
	\caption{The graph of $F$ versus $\rho$ in (\ref{R-roots-rho}) for $c^2 < 4b$ (left) and for $c^2 > 4b$, $c < 0$, and $b > 0$ (right).} \label{figF}
\end{figure}

\begin{figure}[hb!]
	\includegraphics[width=0.48\textwidth]{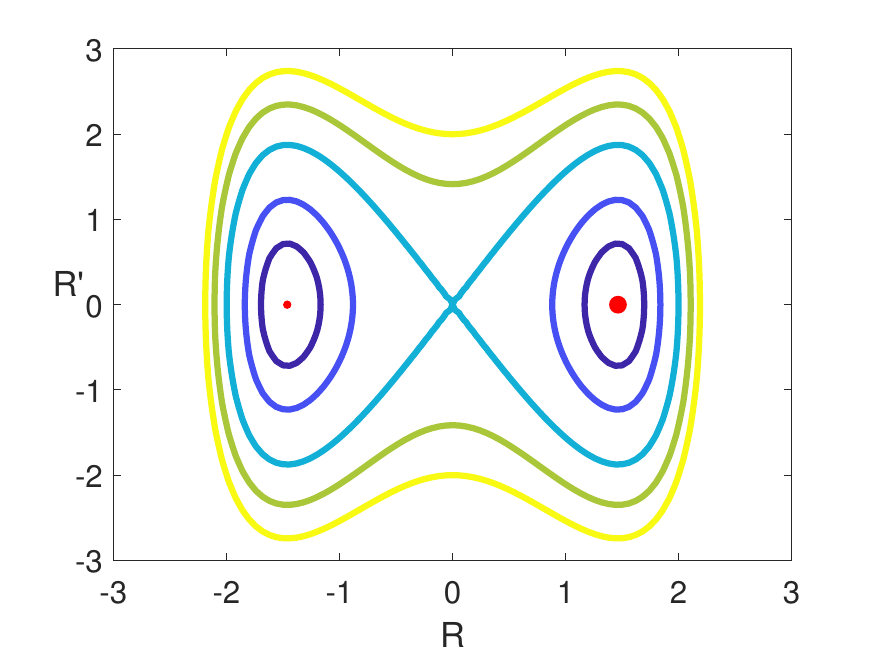}
	\includegraphics[width=0.48\textwidth]{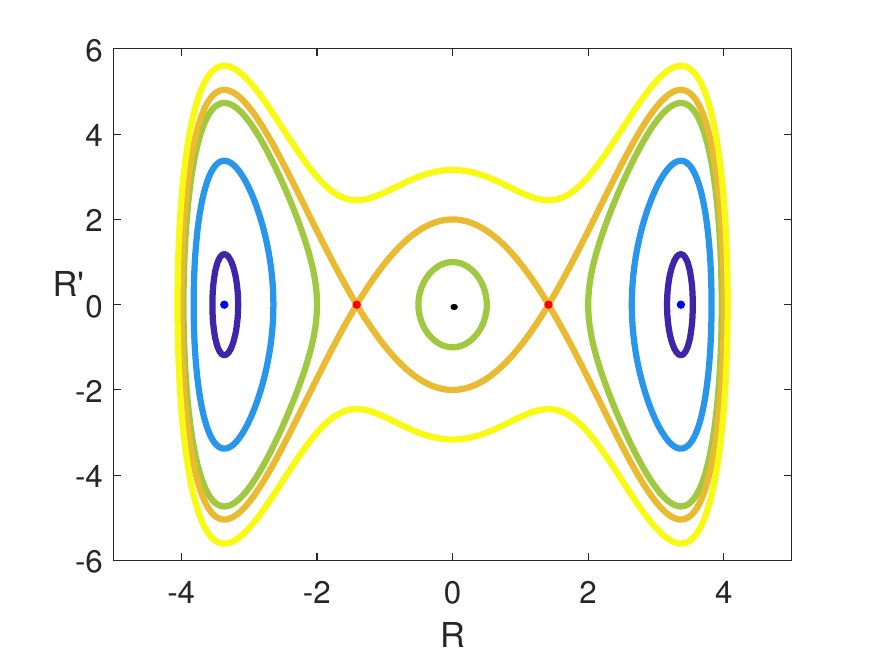}
	\caption{Phase portrait in the phase plane $(R,R')$ for $c^2 < 4b$ (left) and for $c^2 > 4b$, $c < 0$, and $b > 0$ (right).} \label{figPlane}
\end{figure}

Introducing again $\rho := \frac{1}{2} R^2$ and abusing notations for $F$, we can rewrite (\ref{R-roots}) in the form
\begin{equation}
\label{R-roots-rho}
F(\rho) = \frac{1}{2} \rho^3 + 2c \rho^2 + 2 (c^2 - 4b) \rho.
\end{equation}
One root of the cubic polynomial $F(\rho)$ is at zero and the other two roots are given by 
\begin{equation}
\label{rho-roots}
\rho_{\pm} = -2c \pm 4 \sqrt{b}. 
\end{equation}
The graph of $F$ versus $\rho$ is
shown in Fig. \ref{figF} for $c^2 < 4b$ (left) and for $c^2 > 4b$, $c < 0$, and $b > 0$ (right). These two cases correspond to $\rho_- < 0 < \rho_+$ and $0 < \rho_- < \rho_+$, respectively. The other cases of $c^2 > 4b$ and $c > 0$ 
or $c^2 > 4b$, $c < 0$, and $b < 0$ 
correspond to either $\rho_- < \rho_+< 0$ or complex-conjugate $\rho_{\pm}$, so that $F(\rho) > 0$ for $\rho > 0$.

If $c^2 < 4b$, sign-definite periodic solutions exist for $d \in (d_-,0)$,
where $d_- := \frac{1}{4} \min\limits_{\rho \in [0,\infty)} F(\rho)$ (see the left panel of Fig. \ref{figF}).
As $d \to d_-$, the family of periodic solutions degenerates
to the constant-amplitude solution. As $d \to 0$, the family of periodic solutions approaches the solitary wave satisfying $R(x) \to 0$
as $|x| \to \infty$, which corresponds to the exact solution 
(\ref{soliton}).
Sign-indefinite periodic solutions exist 
for $d \in (0,d_+)$ and $d \in (d_+,\infty)$, where 
$d_+ := \frac{1}{4} \max\limits_{\rho \in (-\infty,0]} F(\rho)$. It should be noted that $R(x)$ in the quadrature (\ref{R-quadrature}) is allowed to be negative but both positive and negative values of $R$ correspond to positive values of $\rho = \frac{1}{2} R^2$. Also note that the local maximum point $d_+$ occurring for $\rho \in (-\infty,0)$ affects the analytical representation of the periodic solutions by either (\ref{3.14}) or (\ref{3.20}) but does not change the qualitative 
behavior of $R$. A phase portrait for the quadrature (\ref{R-quadrature}) with $c^2 < 4b$ is shown on the phase plane $(R,R')$ in Fig. \ref{figPlane} (left).

If $c^2 > 4b$ and either $c > 0$ or $c < 0$ and $b < 0$,
then $F(\rho) > 0$ for $\rho > 0$. Sign-indefinite periodic solutions exist for every $d > 0$ but no other bounded
periodic solutions exist. These solutions are very similar to those for $c^2 < 4b$ 
with $d \in (0,\infty)$, therefore, we will not consider examples 
of such periodic solutions for these parameter ranges.

If $c^2 > 4b$, $c < 0$, and $b > 0$,
sign-definite periodic solutions exist for $d \in (d_-,d_+)$, and sign-indefinite periodic solutions exist for $d \in (0,d_+)$ and $d \in (d_+,\infty)$ (see the right panel of Fig. \ref{figF}).
When $d \to d_-$, the family of sign-definite periodic solutions degenerates
to the constant-amplitude solution. As $d \to d_+$, the family of sign-definite periodic solutions approaches the solitary wave satisfying $R(x) \to R_0$ 
as $|x| \to \infty$ with $R_0 > 0$ being the root of $F(R) = 4d$ with $F(R)$ given by (\ref{R-roots}). The family of sign-indefinite periodic solutions approaches the kink solution 
satisfying $R(x) \to \pm R_0$ as $x \to \pm \infty$. Phase portrait for the quadrature (\ref{R-quadrature}) with $c^2 > 4b$, $c < 0$, and $b > 0$ is shown on the phase plane $(R,R')$ in Fig. \ref{figPlane} (right).

If $a = 0$, one root of $Q$ in (\ref{3.9}) is zero. The other three roots are given by the intersection of the graph of $F(\rho)$ given by (\ref{R-roots-rho}) with the constant level $4d$. In the remainder of this section we study the two cases (i) $c^2<4b$ and (ii) $c^2 > 4b$, $c<0$, and $b > 0$. In each case we give exact analytical expressions for the periodic wave solutions and create representative figures of the Floquet spectrum in the KN spectral problem (\ref{lax-1-intro}) using the numerical Hill's method \cite{FFHM,JohnZumb}. 
The connection formula $\Lambda = \pm 4 i \sqrt{P(\lambda)}$ allows us to study the modulational stability or instability of the periodic standing waves in the spectral problem (\ref{dnls-spectral}).

\subsection{Case: $c^2 < 4b$}

If $d \in (d_-,0)$ (see the left panel of Fig. \ref{figF}), then the roots of $Q$ are all real and ordered as 
$$
u_4 < 0 = u_3 < u_2 < u_1.
$$
The exact analytical expression for the periodic wave solutions is given by  (\ref{3.14}) for $\rho$ in $[u_2,u_1]$. The period of the periodic wave is 
$L = 2 K(k) \nu^{-1}$. The roots of $P(\lambda)$ in (\ref{3.22}) form two quadruplets of complex-conjugate eigenvalues in (\ref{configuration-1}) with
$\alpha_1 = \alpha_2$ in (\ref{3.27}).

\begin{figure}[htb!]
	\centering
	\begin{subfigure}{\textwidth}
	$\phantom{txt}$\includegraphics[width=5cm,height = 3.5cm]{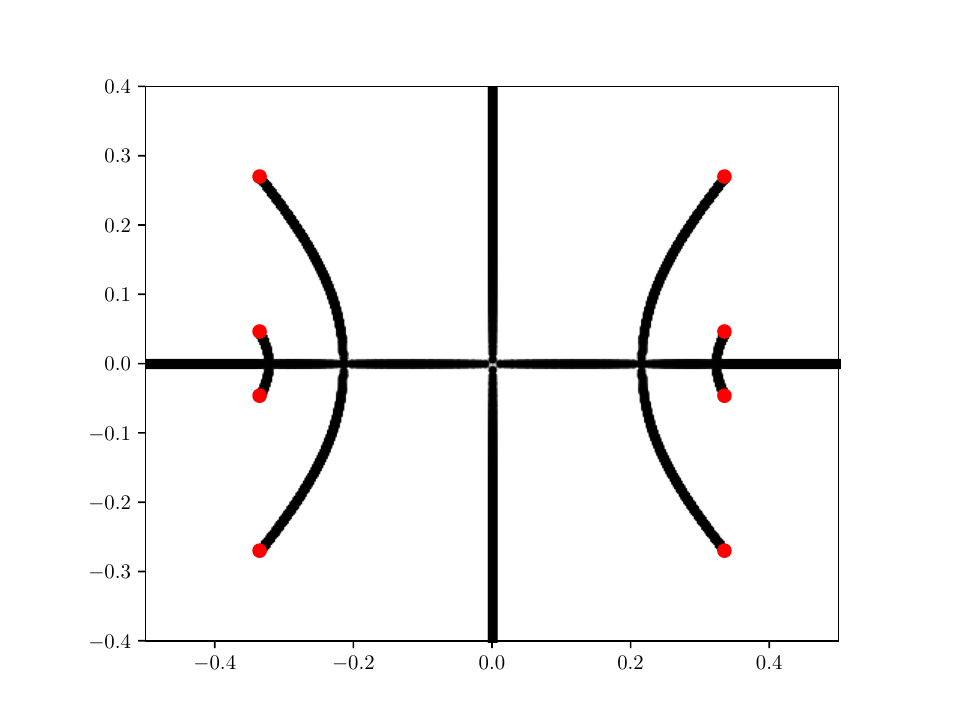}
	\hfil
	\includegraphics[width=6cm,height = 3.5cm]{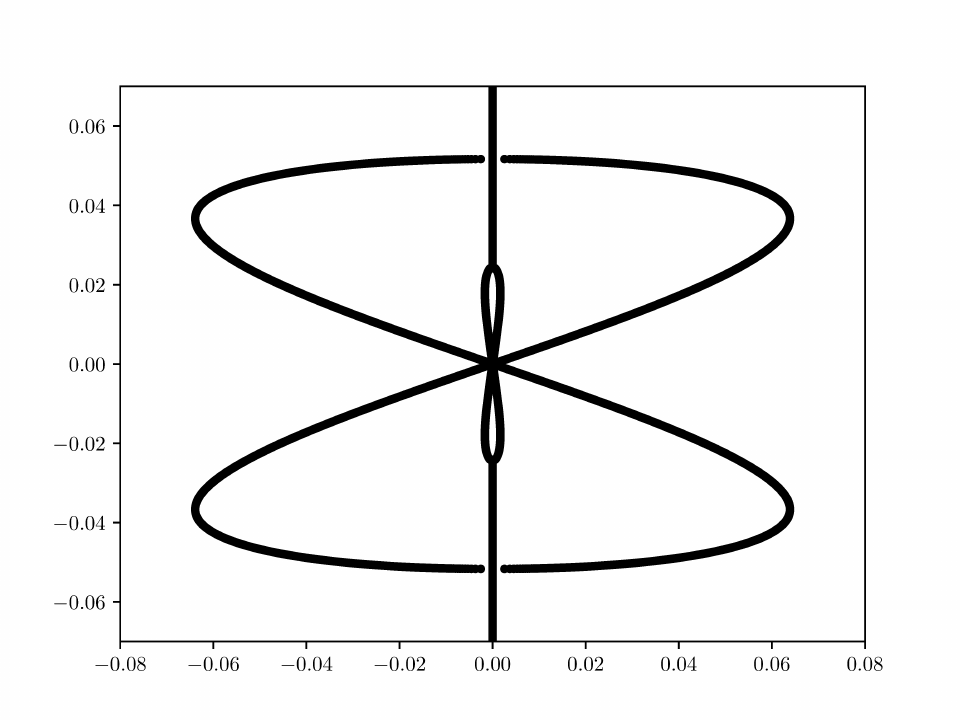}
		\caption{$u_1 = 0.2,$ $u_2 = 0.1$, $u_3 = 0$, $u_4 = -0.9$.} 
	\end{subfigure}
	\begin{subfigure}{\textwidth}
	\includegraphics[width=6cm,height = 3.5cm]{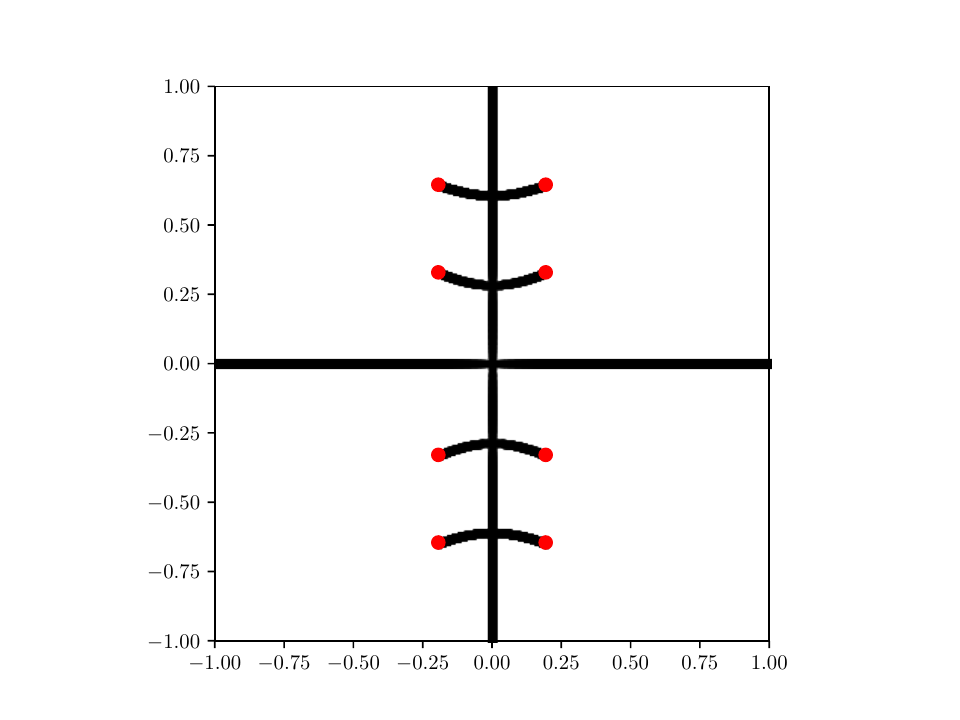}\hfil
	\includegraphics[width=6cm,height = 3.5cm]{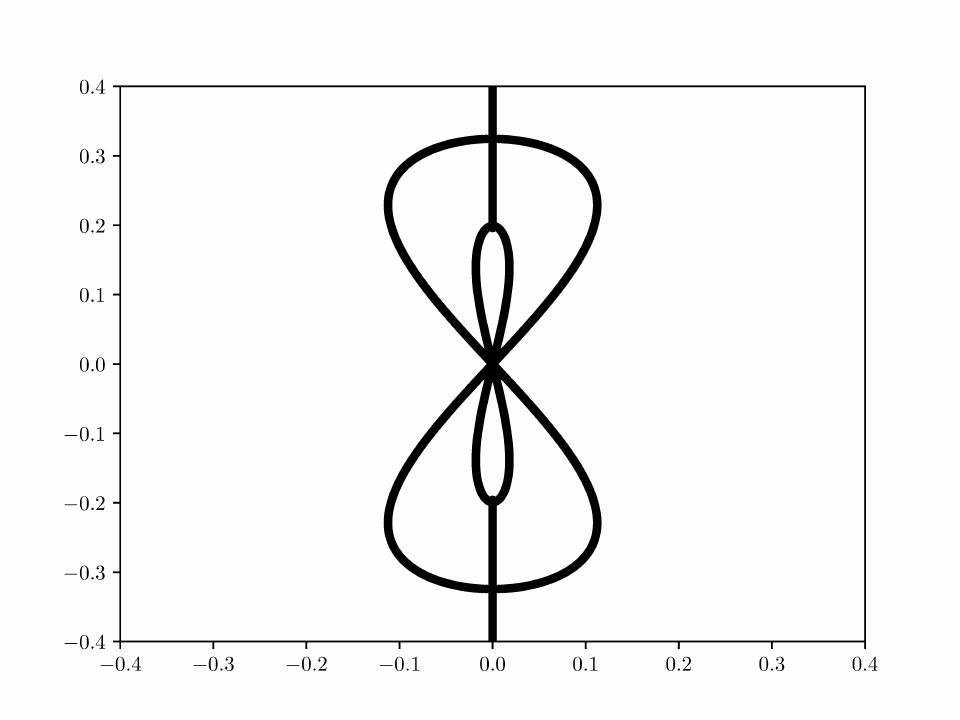} 
		\caption{$u_1 = 1.9,$ $u_2 = 0.2$, $u_3 = 0$, $u_4 = -0.3$. }
	\end{subfigure}
	\begin{subfigure}{\textwidth}
	\includegraphics[width=6cm,height = 3.5cm]{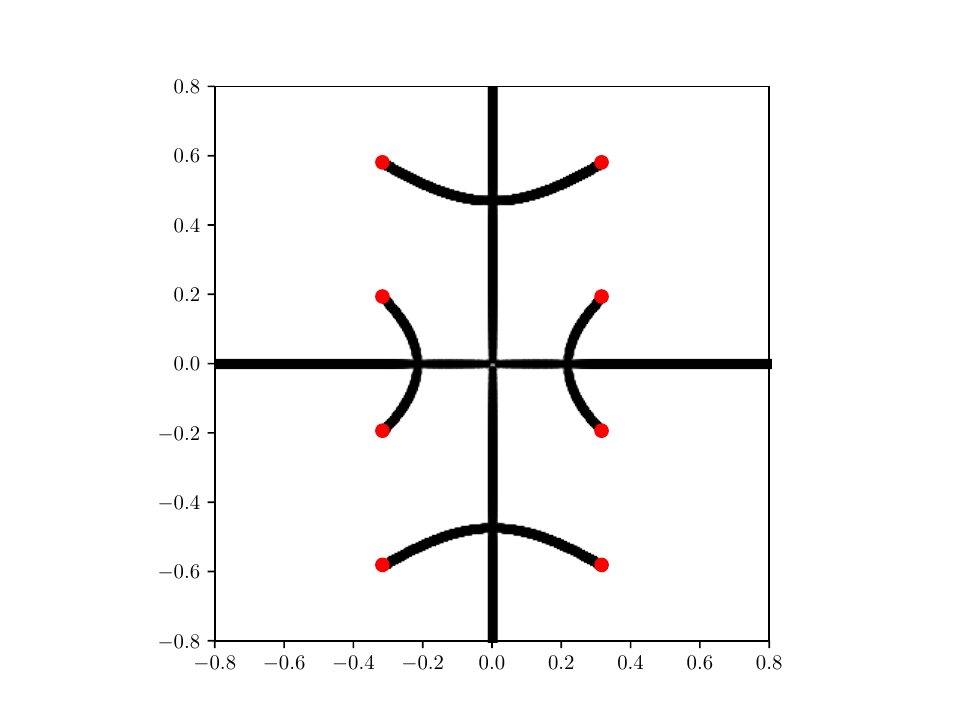}\hfil
	\includegraphics[width=6cm,height = 3.5cm]{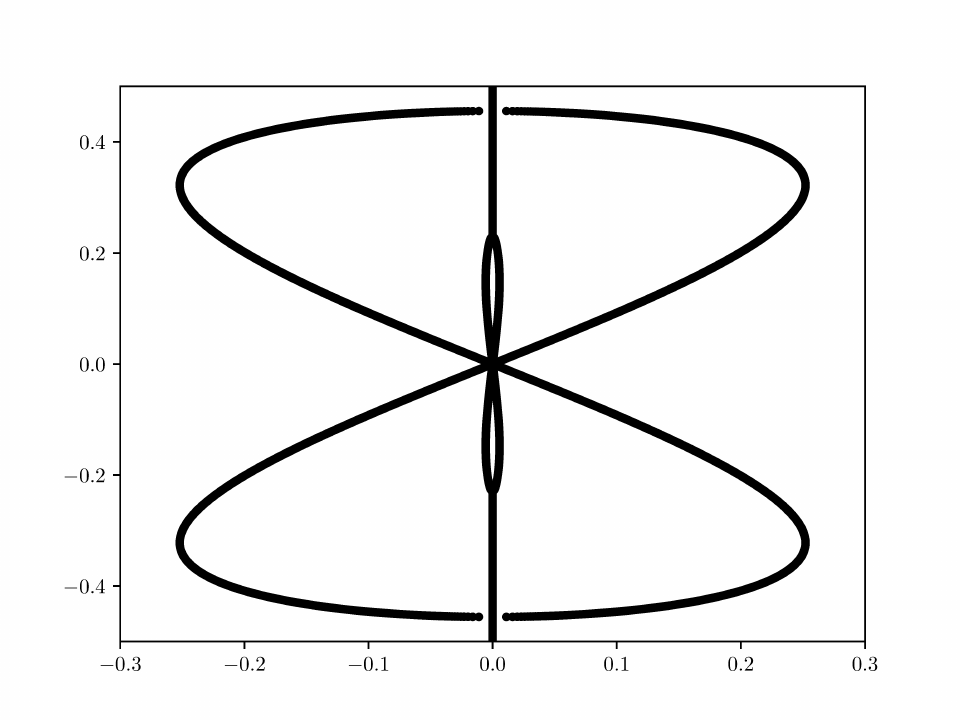} 
		\caption{$u_1 = 1.2,$ $u_2 = 0.3$, $u_3 = 0$, $u_4 = -0.8$. }
	\end{subfigure}
	\begin{subfigure}{\textwidth}
	\includegraphics[width=6cm,height = 3.5cm]{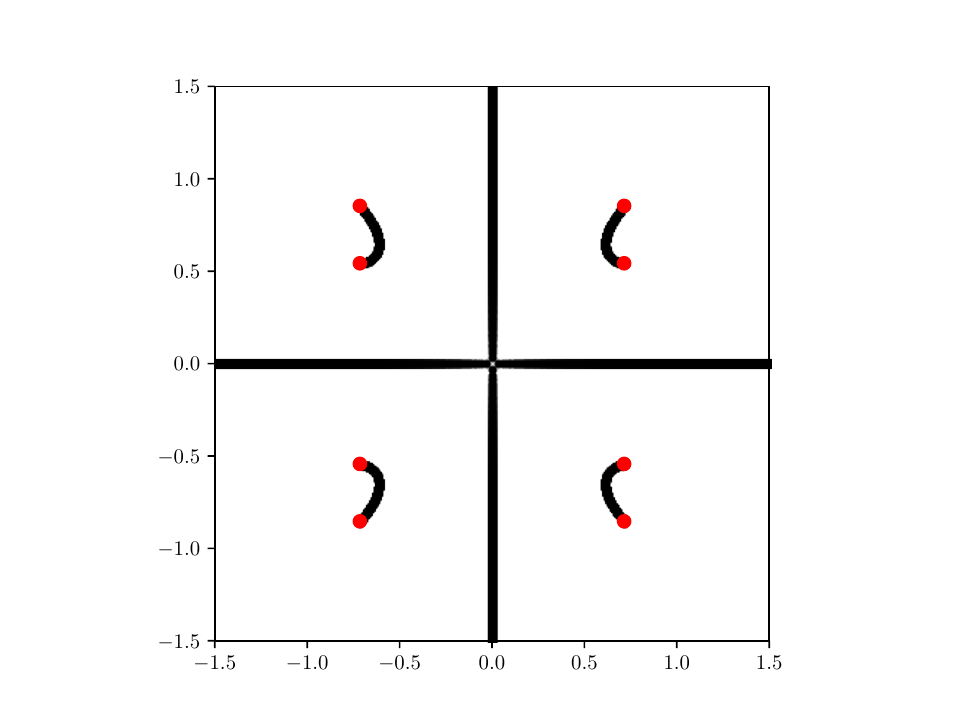}\hfil
	\includegraphics[width=6cm,height = 3.5cm]{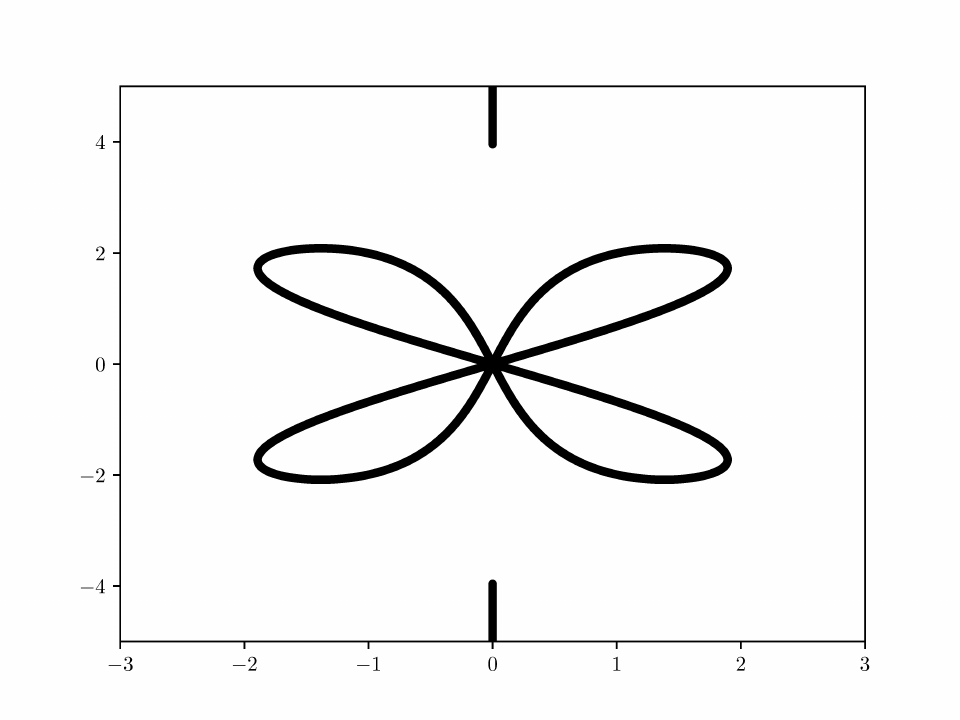}
		\caption{$u_1 = 3.9,$ $u_2 = 0.193012$, $u_3 = 0$, $u_4 = -4.090301$. }
	\end{subfigure}
	\caption{Numerical computations of the Floquet spectrum in the KN spectral problem (left) and stability spectrum (right) on the complex $\lambda$ and $\Lambda$ planes (Re. \textit{vs.} Im.) for the four representative cases found for $c^2<4b$ and $d \in (d_-,0)$. Eight red dots on the left panel represent roots of the polynomial $P(\lambda)$.
	\label{figure1-4}}
\end{figure}

\begin{figure}[htb!]
	\begin{center}
		\begin{subfigure}{\textwidth}
			\includegraphics[width=6cm,height=3.8cm]{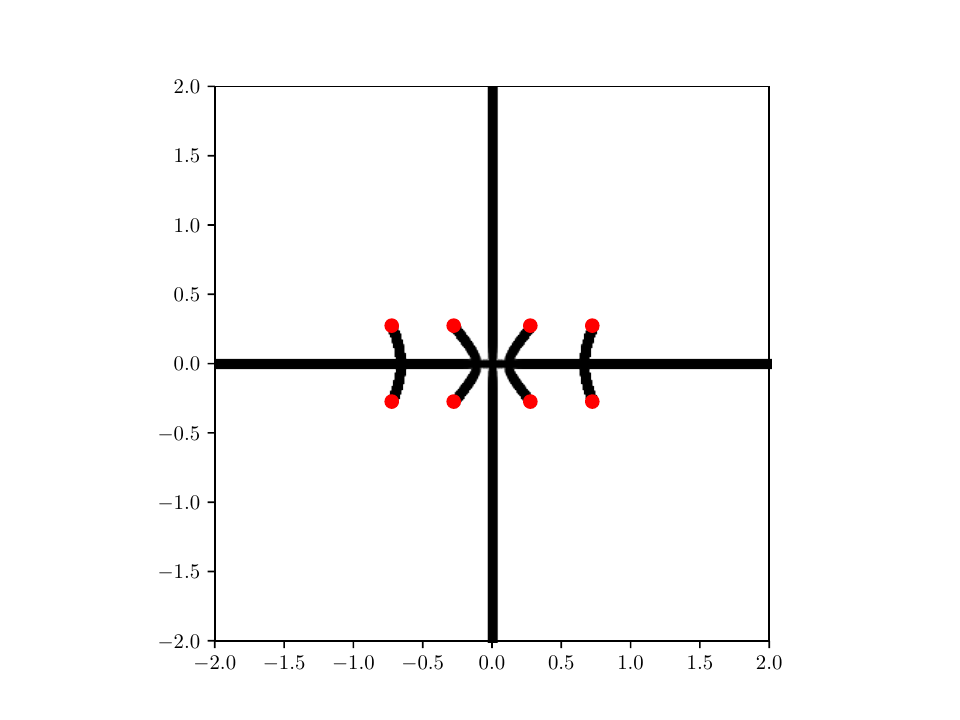}\hfil
			\includegraphics[width=6cm,height=3.8cm]{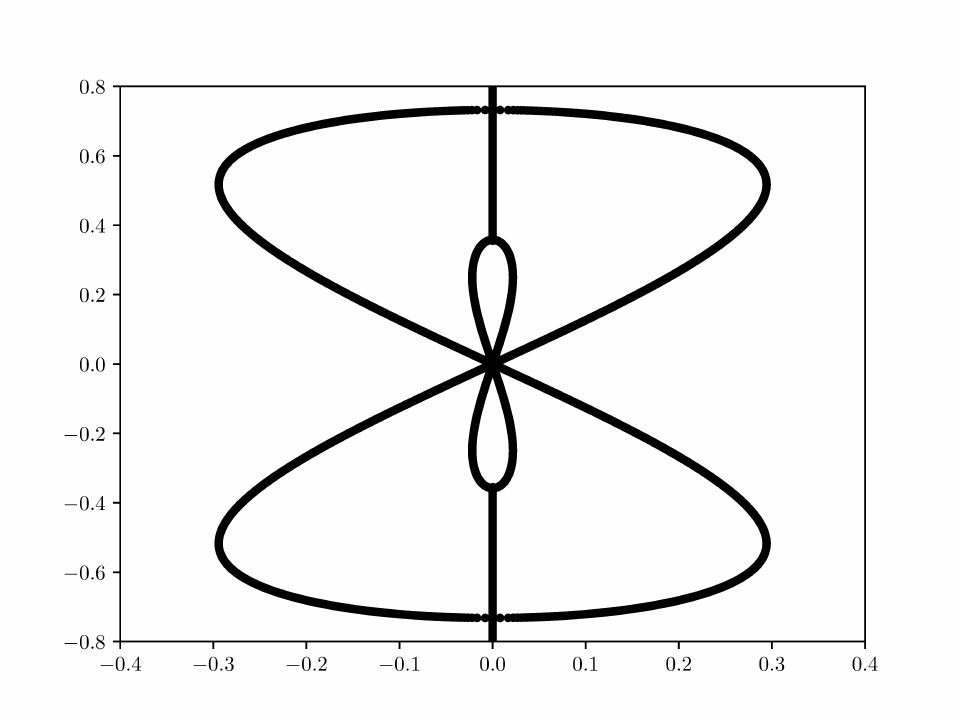} 
			\caption{$u_1 = 0.6$, $u_2 = 0,$ $u_3 = -0.4$, $u_4 = -2$.} 
		\end{subfigure}
		\begin{subfigure}{\textwidth}
			\includegraphics[width=6cm,height=3.8cm]{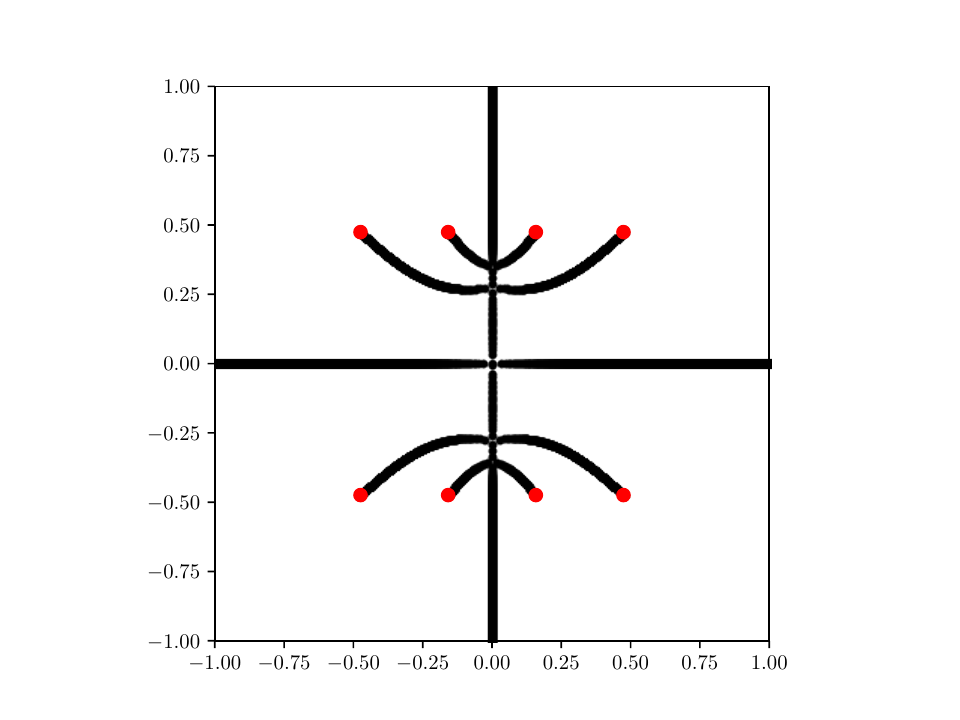}\hfil
			\includegraphics[width=6cm,height=3.8cm]{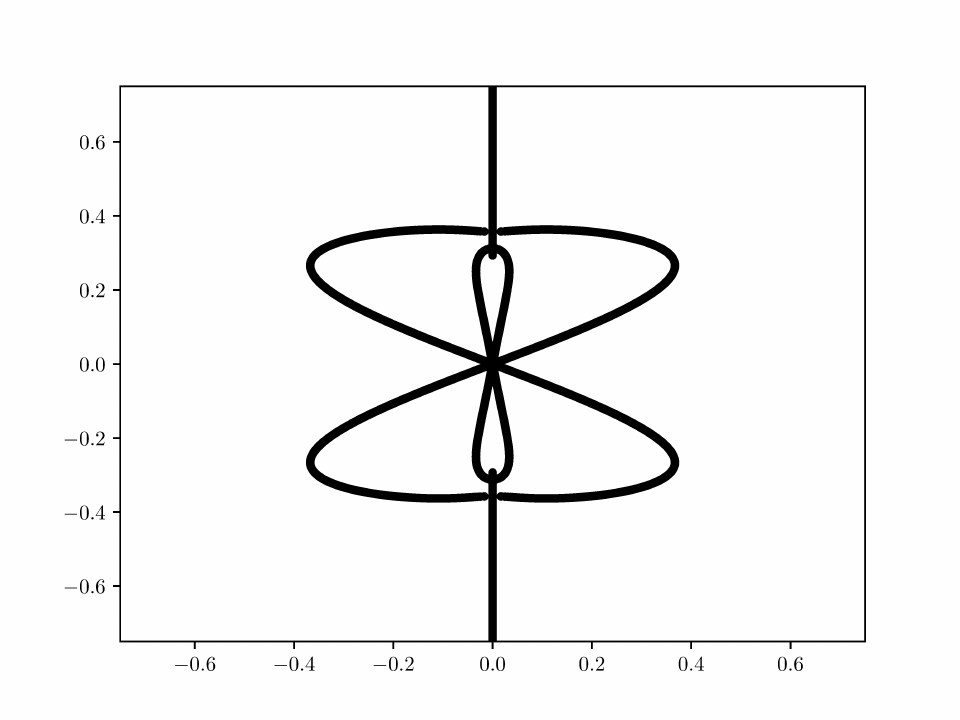} 
			\caption{$u_1 = 1.8$, $u_2 = 0,$ $u_3 = -0.2$, $u_4 = -0.8$.} 
		\end{subfigure}
		\begin{subfigure}{\textwidth}
			\includegraphics[width=6cm,height=3.8cm]{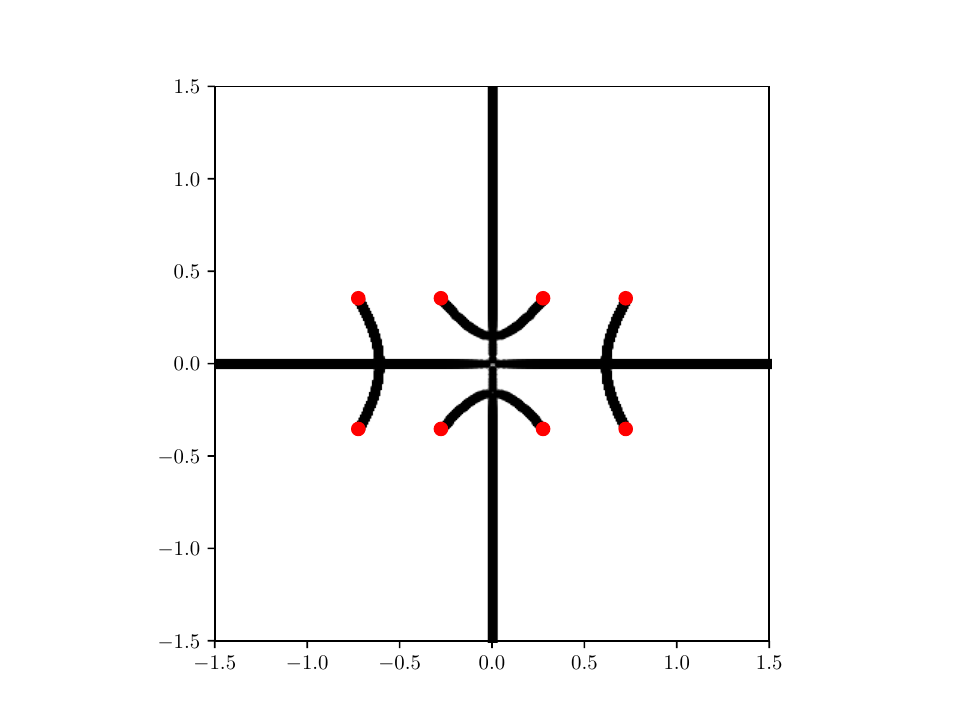}\hfil
			\includegraphics[width=6cm,height=3.8cm]{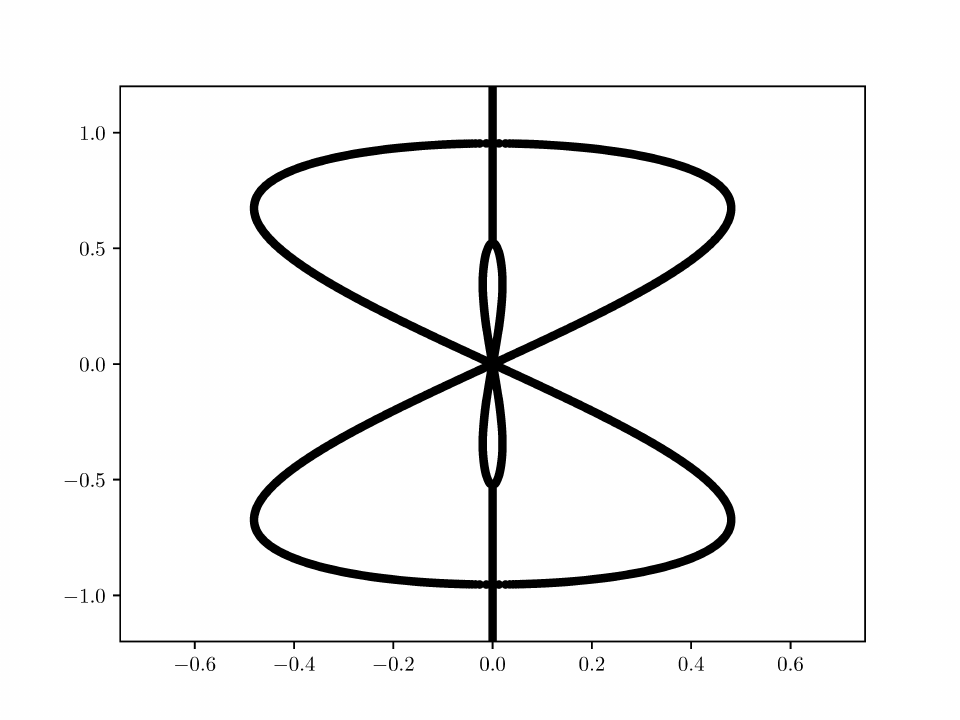}
			\caption{$u_1 = 1$, $u_2 = 0,$ $u_3 = -0.4$, $u_4 = -2$.} 
		\end{subfigure}
		\begin{subfigure}{\textwidth}
			\includegraphics[width=6cm,height=3.8cm]{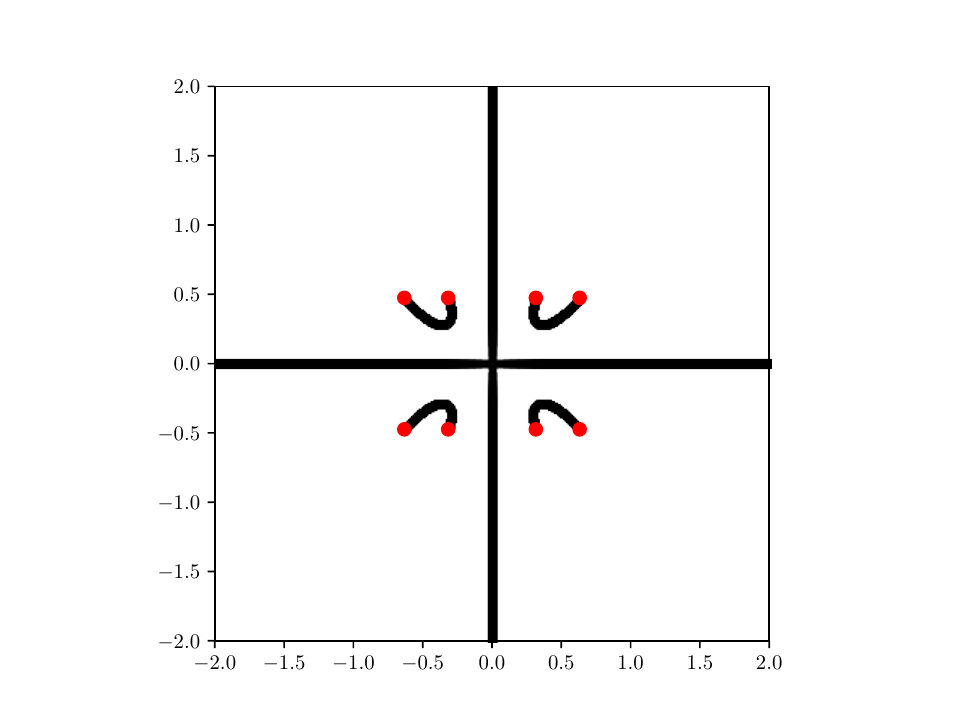}\hfil
			\includegraphics[width=6cm,height=3.8cm]{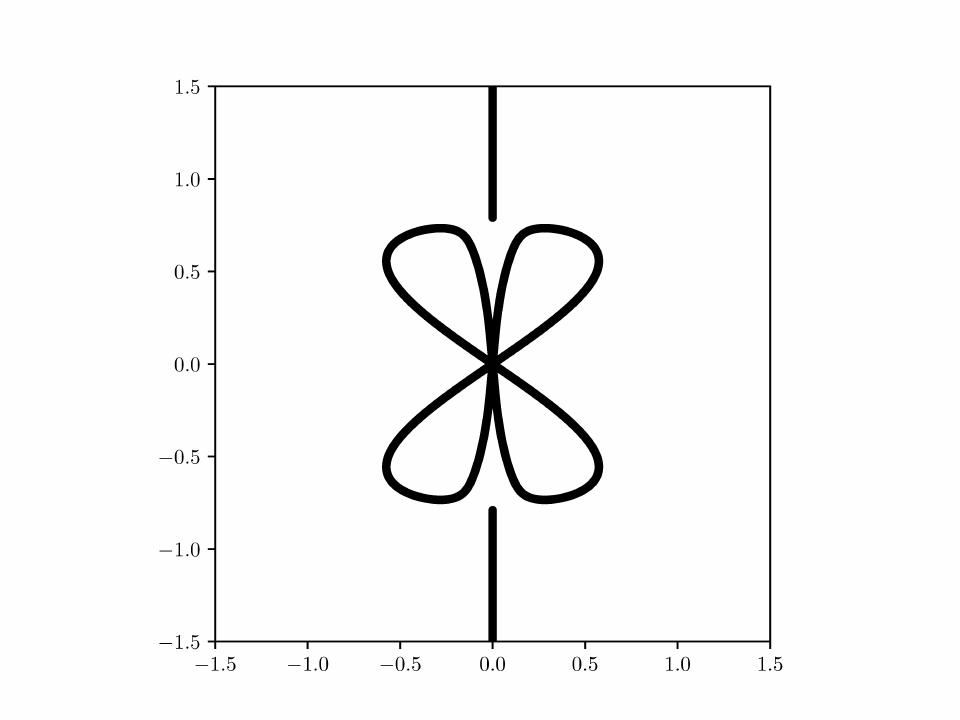}
			\caption{$u_1 = 1.8$, $u_2 = 0,$ $u_3 = -0.2$, $u_4 = -1.8$.} 
		\end{subfigure}
		\caption{The same as Figure \ref{figure1-4} but for $c^2<4b$ and $d \in (0, d_+)$.
			\label{figure5-8}}
	\end{center}
\end{figure}

As $d \to d_-$, we have $u_2 \to u_1$ and $\beta_2 \to 0$ in (\ref{3.27}),
hence one quadruplet persists in the limit but the other
coalesces on the real axis. This corresponds to the constant-amplitude wave. As $d \to 0$, we have $u_2 \to u_3 = 0$ and $\beta_2 \to \beta_1$ in (\ref{3.27}), hence two quadruplets coalesce in the complex plane 
to a single quadruplet. This corresponds to the solitary wave (\ref{soliton}). 

For $d \in (d_-,0)$, we find four typical configurations for the Floquet spectrum of the KN spectral problem (\ref{lax-1-intro}) shown in Figure \ref{figure1-4} (left): (a) with four spectral bands intersecting the real axis, (b) with four spectral bands intersecting the imaginary axis, (c) with two spectral bands intersecting the real axis and two spectral bands intersecting the imaginary axis, and (d) with no  intersections with the real or imaginary axes. The first three cases lead to a double figure-8 of the stability spectrum 
shown in Fig. \ref{figure1-4} (right): one figure-8 is embedded within another. The last case leads to a butterfly figure. Both figure-8 and the butterfly figure were observed for the periodic standing waves of nontrivial phase in the NLS equation \cite{DS}.  Each case leads to the modulational instability of the periodic standing waves according to Definition \ref{def-stab}.

If $d \in (0,d_+)$,
the roots of $Q(\rho)$ are all real and  ordered by
$$
 u_4 < u_3 < 0 = u_2 < u_1.
$$
The exact analytical expression for the periodic wave solutions is given by  (\ref{3.14}) for $\rho$ in $[0,u_1]$. However, the case when $\rho(x)$ may vanish corresponds to the case of the sign-indefinite $R(x)$. 
If $u_2 = 0$ is used in the expression (\ref{3.14}),
the expression can be written as
\begin{equation}\label{3.14-simple}
\rho(x) =
\frac{u_1 {\rm cn}^2(\nu x;k)}{1 + |u_4|^{-1} u_1 {\rm sn}^2 (\nu x;k)}.
\end{equation}
Extracting the square root analytically yields the exact 
expression for the periodic wave solutions,
\begin{equation}\label{3.14-R}
R(x) =
\frac{\sqrt{2u_1} {\rm cn}(\nu x;k)}{\sqrt{1 + |u_4|^{-1} u_1 {\rm sn}^2 (\nu x;k)}}.
\end{equation}
The period of the periodic wave is now $L = 4 K(k) \nu^{-1}$ (which is double compared to the case of sign-definite solutions).
The roots of $P(\lambda)$ in (\ref{3.22}) form two quadruplets
of complex-conjugate eigenvalues in (\ref{configuration-1}) with
$\beta_1 = \beta_2$ in (\ref{3.27}).

For $d \in (0,d_+)$, we find four typical configurations for the Floquet spectrum in the $\lambda$ and $\Lambda$ planes shown in Figure \ref{figure5-8}. Each case is similar to one in Fig. \ref{figure1-4}.

As $d \to d_+$, we have $u_3 \to u_4$ and $\alpha_2 \to 0$ in (\ref{3.27}),
hence one quadruplet persists in the limit but the other 
coalesces on the imaginary axis. The sign-indefinite periodic solution
continues for $d \in (d_+,\infty)$ but now corresponds
to the case of two real roots of $Q(\rho)$ with $0 = u_2 < u_1$
and a pair of complex-conjugate roots $u_{3,4} = \gamma \pm i \eta$.
The exact analytical expression for the periodic wave is given by (\ref{3.20}) for $\rho$ in $[0,u_1]$. Again, the case when $\rho(x)$ may vanish corresponds to the case of the sign-indefinite $R(x)$. 
If $u_2 = 0$ is used in the expression (\ref{3.20}),
the expression can be written as
\begin{equation}\label{3.20-simple}
\rho(x) = u_1 \delta \frac{1+{\rm cn} (\mu x;k)}{1+\delta +(\delta-1){\rm cn} (\mu x;k)}.
\end{equation}
Extracting the square root analytically yields the exact 
expression for the periodic wave solutions, 
\begin{equation}\label{3.20-R}
R(x) = \frac{\sqrt{2u_1 \delta} {\rm cn} (\frac 12 \mu x;k)} {\sqrt{\delta {\rm cn}^2 (\frac 12 \mu x;k) +
{\rm sn}^2 (\frac 12 \mu x;k) {\rm dn}^2 (\frac 12 \mu x;k)}}.
\end{equation}
The period of the periodic wave is  $L = 8 K(k) \mu^{-1}$ (double compared to the case of sign-definite solutions).
The roots of $P(\lambda)$ in (\ref{3.22}) form one quadruplet
of complex-conjugate eigenvalues $(\pm \lambda_1,\pm \bar{\lambda}_1)$ and
two pairs of purely imaginary eigenvalues $(\pm i \beta_3,\pm i \beta_4)$ as in (\ref{configuration-3}) and (\ref{a3.16}).

\begin{figure}[hb!]
	\centering
	\begin{subfigure}[b]{\textwidth}
	\includegraphics[width=6cm,height = 4cm]{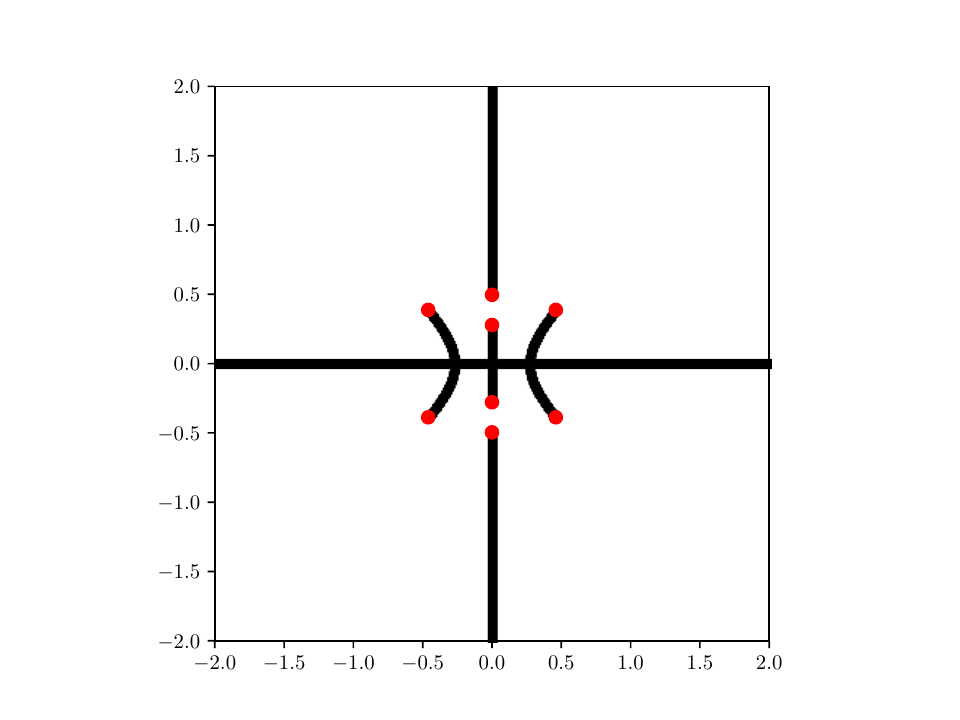}\hfil
	\includegraphics[width=6cm,height = 4cm]{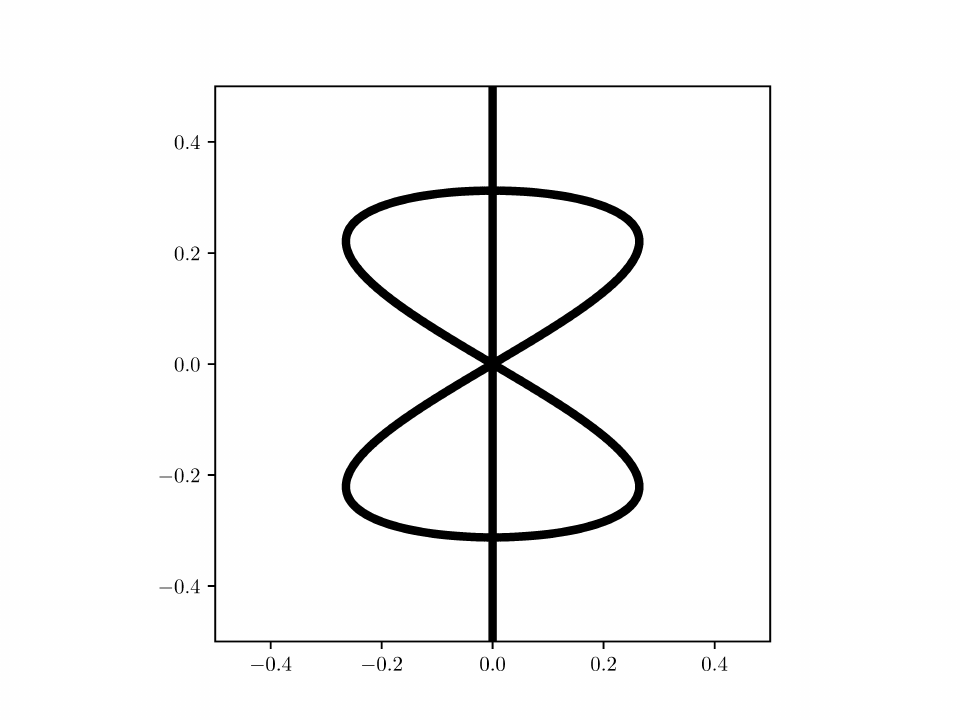} 
	\caption{$u_1 = 1.2$, $u_2 = 0$, $u_3 = -0.4-0.2i$, $u_4 = -0.4+0.2i$.}
	\end{subfigure}
	\begin{subfigure}[b]{\textwidth}
	\includegraphics[width=6cm,height = 4cm]{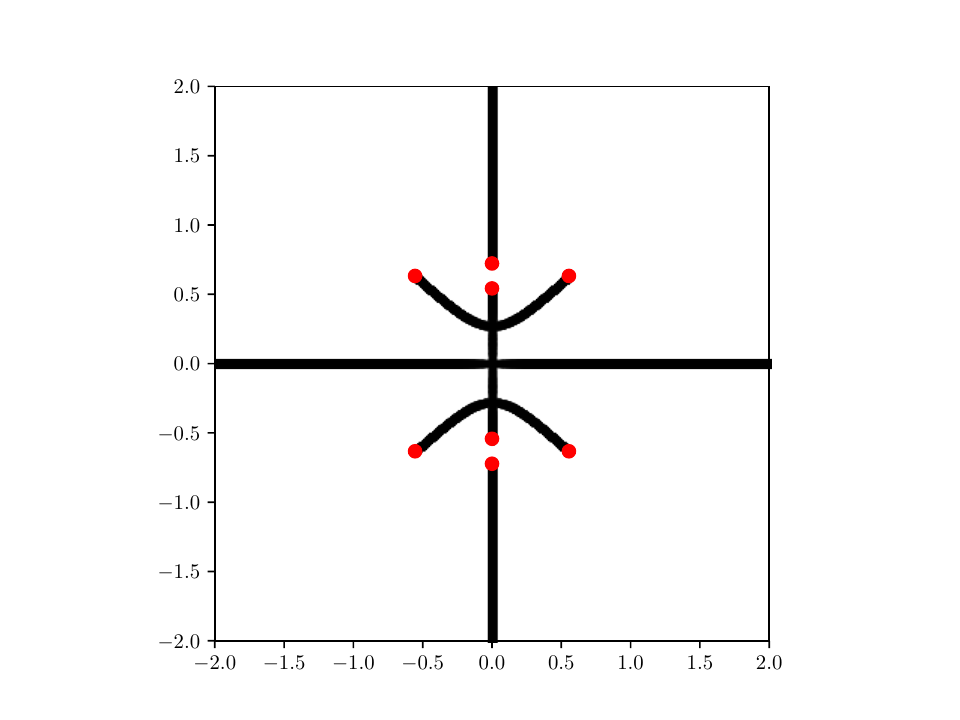}\hfil
	\includegraphics[width=6cm,height = 4cm]{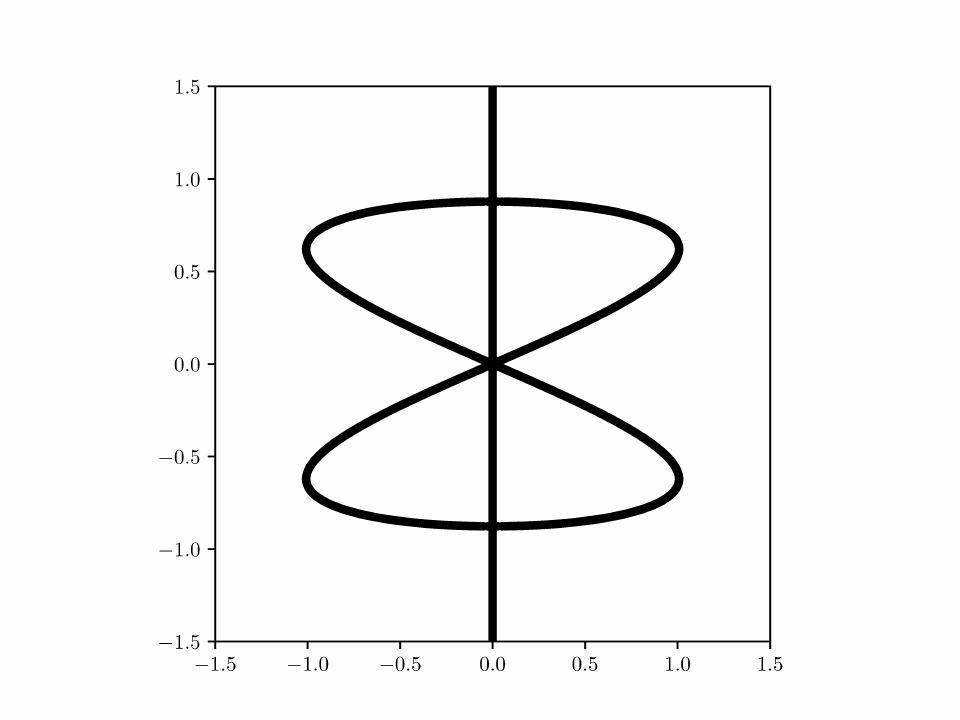}
	\caption{$u_1 = 3.2$, $u_2 = 0$, $u_3 = -0.6+0.2i$, $u_4 = -0.6-0.2i$.}
	\end{subfigure}
	\begin{subfigure}[b]{\textwidth}
	\includegraphics[width=6cm,height = 4cm]{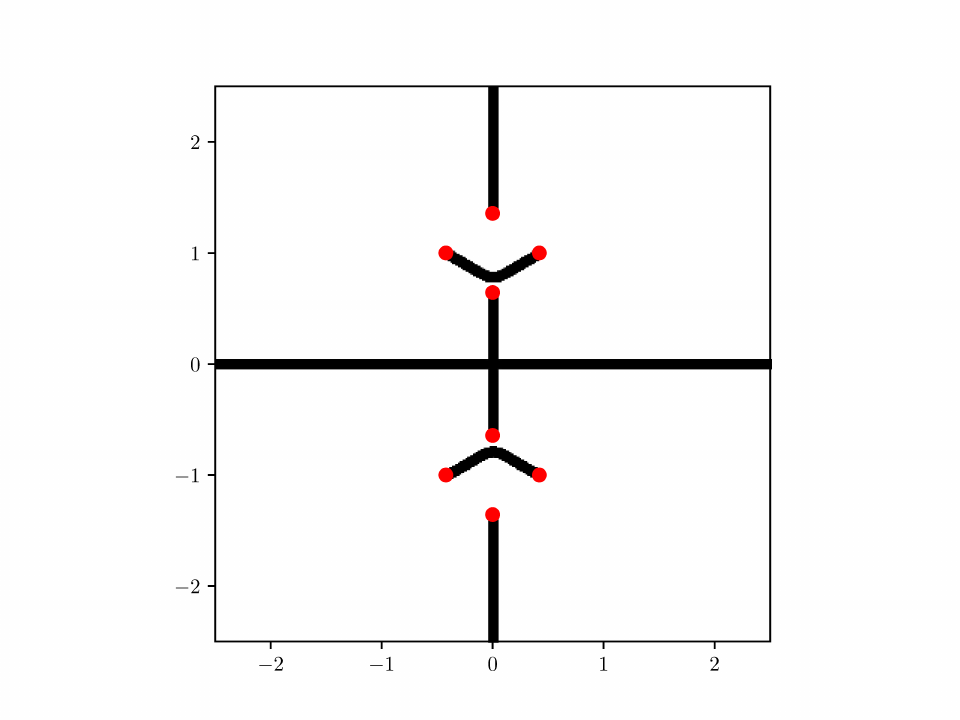}\hfil
	\includegraphics[width=5cm,height = 4cm]{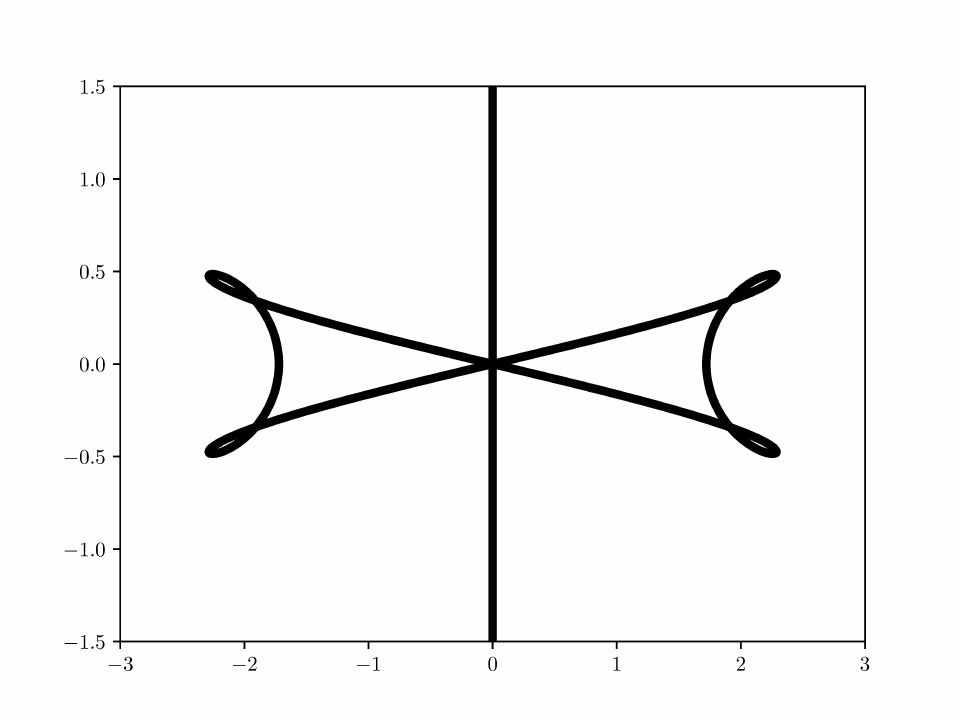}
	\caption{$u_1 = 8$, $u_2 = 0$, $u_3 = -0.1+0.6i$, $u_4 = -0.1-0.6i$.}
	\end{subfigure}
	\caption{The same as Figure \ref{figure1-4} but for $c^2<4b$ and $d \in (d_+, \infty)$.
	\label{figure9}}
\end{figure}

For $d \in (d_+,\infty)$, we find three typical configurations for the Floquet spectrum in the KN spectral problem (\ref{lax-1-intro}) shown in Figure \ref{figure9} (left). All three have two spectral gaps on the imaginary axis. The difference between the three configurations 
is as follows: (a) two spectral bands intersect the real axis, (b) two spectral bands intersect the imaginary axis in the inner spectral band on the imaginary axis, and (c) two spectral bands intersect the imaginary axis in the spectral gaps. The stability spectrum in (a) and (b) cases represent a single figure 8 shown on Fig. \ref{figure9} (right). The stability spectrum in (c) is a novel shape, which was not seen for the periodic standing waves in the NLS equation in \cite{DS}. The gap on the imaginary axis satisfies the stability condition (\ref{stab-1}). However, each periodic wave has a complex band connected to the origin and hence is modulationally unstable according to Definition \ref{def-stab}.

\subsection{Case: $c^2 > 4b$, $c < 0$, and $b> 0$}

If $d \in (d_-,0)$ (see the right panel of Fig. \ref{figF}), then the roots of $Q$ are all real and ordered as 
$$
u_4 < 0 = u_3 < u_2 < u_1.
$$
The exact analytical expression for the periodic wave solutions is given by  (\ref{3.14}) for $\rho$ in $[u_2,u_1]$ with the period $L = 2 K(k) \nu^{-1}$.
The roots of $P(\lambda)$ in (\ref{3.22}) form two quadruplets
of complex-conjugate eigenvalues in (\ref{configuration-1}) with
$\alpha_1 = \alpha_2$ in (\ref{3.27}). 
This case leads to similar figures as those in Figure \ref{figure1-4}. All periodic standing waves are modulationally unstable. 

When $d \to d_-$, we have $u_2 \to u_1$ and one quadruplet coalesces on the real axis. This corresponds to the constant-amplitude wave. When $d \to 0$, we have $u_4 \to u_3 = 0$ and both quadruplets coalesce on the imaginary axis. 
At this point, the sign-definite periodic solution is continued for $d \in (0,d_+)$ but another sign-indefinite periodic solution arises.

If $d \in (0,d_+)$, then the roots of $Q$ are  ordered as 
$$
u_4 = 0 < u_3 < u_2 < u_1.
$$
As is described above, two periodic solutions coexist: one sign-definite solution is given by (\ref{3.14}) for $\rho \in [u_2,u_1]$ and the other sign-indefinite solution is given by (\ref{3.14a}) for $\rho \in [0,u_3]$. Extracting the square root analytically yields the sign-indefinite solution in the exact form,
\begin{equation}\label{3.14-R-R}
R(x) =
\frac{\sqrt{2u_3} {\rm cn}(\nu x;k)}{\sqrt{1 - u_2^{-1} u_3 {\rm sn}^2 (\nu x;k)}}.
\end{equation}
The roots of $P(\lambda)$ in (\ref{3.22}) for both periodic solutions form four pairs of purely imaginary roots $\{ \pm i \beta_1,  \pm i \beta_2, \pm i \beta_3, \pm i \beta_4\}$ in (\ref{configuration-2}) and (\ref{3.27a}).

For $d \in (0,d_+)$ and for either sign-definite or for sign-indefinite solutions, we find only one typical configuration for the Floquet spectrum in the KN spectral problem (\ref{lax-1-intro}) shown on Figure \ref{figure11} (left). The Floquet spectrum consists of the real axis and the imaginary axis with four spectral gaps. Since there is no Floquet spectrum 
in the spectral gaps in (\ref{stab-2}), the periodic standing waves are spectrally stable by Propositions \ref{prop-stability-real} and \ref{prop-stability}. Indeed, the stability spectrum is on the imaginary axis 
shown on Figure \ref{figure11} (right).

\begin{figure}[htb!]
	\includegraphics[width=0.45\textwidth]{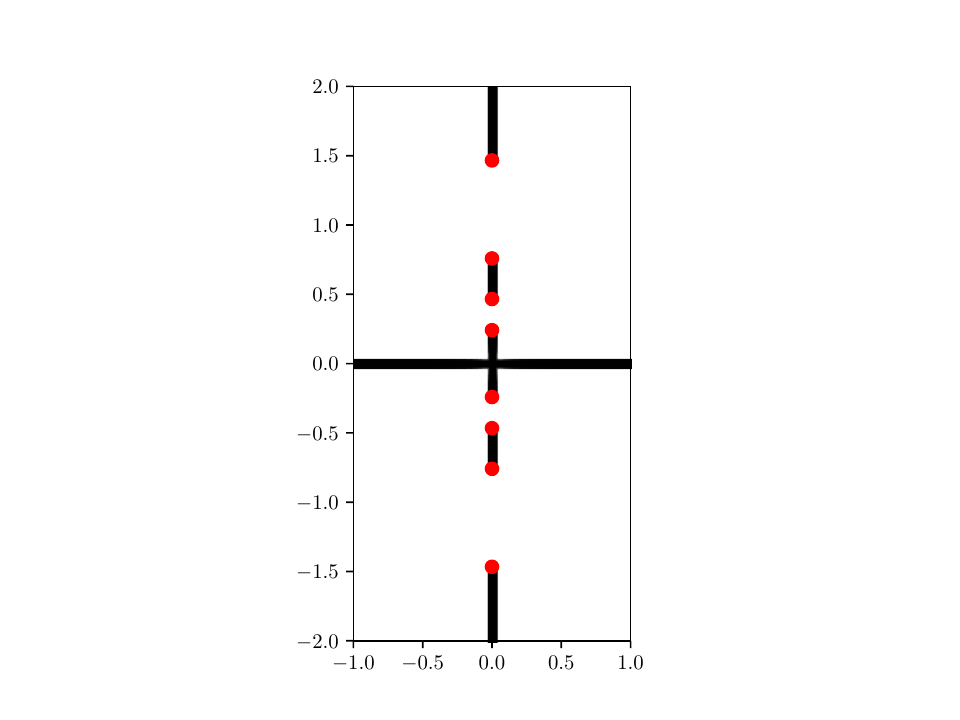}
	\includegraphics[width=0.45\textwidth]{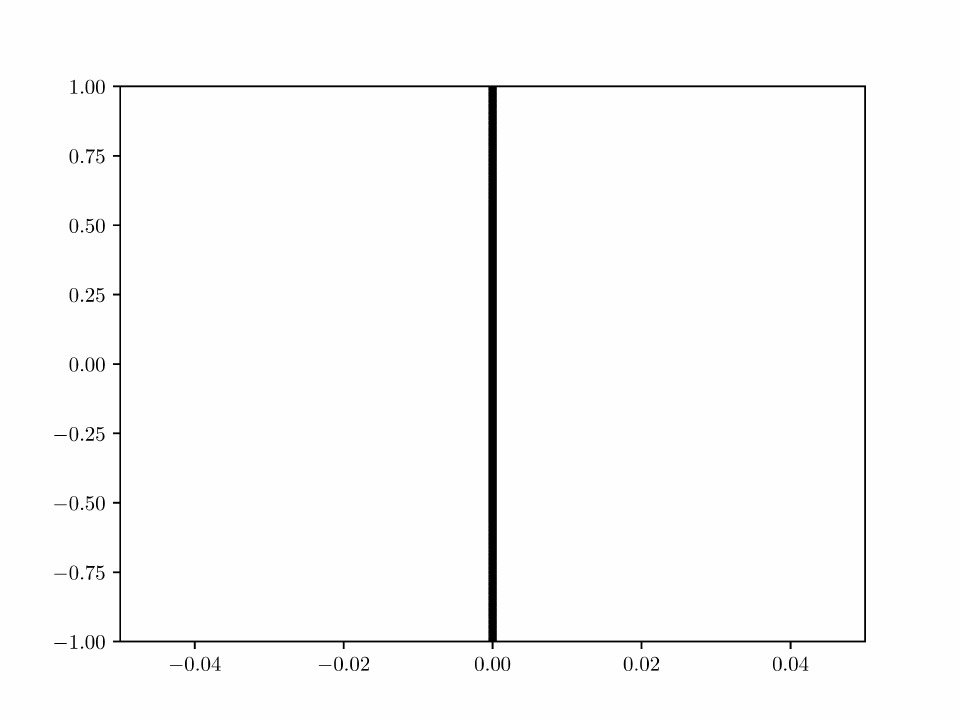}
	\caption{The same as Figure \ref{figure1-4} but for $c^2>4b$, $c<0$, $b>0$ and $d\in (0,d_+)$. The parameters are: $u_1 = 5$, $u_2 = 1$, $u_3 = 0.5$, $u_4 = 0$} \label{figure11}
\end{figure}

When $d \to d_+$, we have $u_3 \to u_2$ and $\beta_2 \to \beta_4$ so that middle spectral bands on the purely imaginary axis coalesce 
with the inner spectral bands, after which the spectral bands re-emerge in the complex plane transversely 
to the imaginary axis. These complex spectral bands intersect the imaginary axis in the inner spectral band on the imaginary axis as seen on Fig. \ref{figure9} (middle).

If $d \in (d_+,\infty)$, the roots of $Q$ can be re-enumerated and ordered as $u_2 = 0  < u_1$ with $u_3 = \bar{u}_4 = \gamma + i \eta$ being complex-conjugate. The exact solution is given by (\ref{3.20-R}). 
The roots of $P(\lambda)$ in (\ref{3.22}) corresponds to one quadruplet of complex eigenvalues $(\pm \lambda_1,\pm \bar{\lambda}_1)$ and two pairs of purely imaginary eigenvalues $\{ \pm i \beta_3, \pm i \beta_4\}$ in (\ref{configuration-3}) and (\ref{a3.16}). This case leads to similar figures as those in Figure \ref{figure9}. All periodic standing waves are modulationally unstable. 

Outcomes of the Floquet and stability spectra for the periodic standing waves in the case $a = 0$ are summarized as follows. 

\vspace{0.2cm}

\centerline{\fbox{\parbox[cs]{1,0\textwidth}{
The only difference between the cases $c^2 < 4b$ and $c^2 > 4b$, $c < 0$, $b > 0$ appears in the narrow interval $d \in (0,d_+)$. For $c^2 < 4b$, there is only one sign-indefinite periodic wave for each $d \in (0,d_+)$ and it is modulationally unstable according to Figure \ref{figure5-8}. For $c^2 > 4b$, $c < 0$, $b > 0$, there are two periodic waves (one is sign-definite and the other one is sign-indefinite) for each $d \in (0,d_+)$; both are spectrally stable according to Figure \ref{figure11}. 
For $d \in (d_-,0)$ and $d \in (d_+,\infty)$, the periodic standing waves between the two cases are similar and the spectral pictures are given on Figures \ref{figure1-4} and \ref{figure9} respectively.}}}

\section{Conclusion}
\label{sec:5}

In this work we have developed the algebraic method of the nonlinearization of linear equations in the Lax pair in order to classify 
all periodic standing waves of the DNLS equation in terms of the location of eight complex eigenvalues of the KN spectral problem. With the assistance of the numerical Hill's method, we have computed the location of the Floquet spectrum in the KN spectral problem. This allowed us to conclude that the periodic standing waves with all eight eigenvalues on the imaginary axis were spectrally (and modulationally) stable, whereas all other periodic standing waves were modulationally unstable. 

We showed these results for the periodic standing waves in the particular case $a = 0$. However, since the eight roots of the polynomial $P(\lambda)$ in (\ref{polynomial}) have similar location for the periodic standing waves in the general case $a \neq 0$, we expect that the same stability conclusions hold for $a \neq 0$. 

A number of new directions in the context of the DNLS equation are opened following this work. Even if the periodic standing waves are modulationally unstable, they can be orbitally stable with respect to periodic perturbations of the same or multiple period, as was explored for the NLS equation in \cite{DU1}. Nonlinear stability analysis for the DNLS equation is an open problem, whereas some results in this direction for the perturbations of the same period were found in \cite{Hakaev}.

Another interesting problem is to locate the Floquet spectrum of the KN spectral problem in the complex $\lambda$ plane analytically. For example, we do not have the analytical proof that the Floquet spectrum on the imaginary axis always have gaps between the spectral gaps as in 
(\ref{stab-1}) and (\ref{stab-2}). There is no proof that there are no other spectral bands of the Floquet spectrum in the open quadrants of the complex plane in addition to those connecting the eight eigenvalues of the algebraic method. Such a proof for the NLS equation was carried out in \cite{DU1}, which may be the starting point for a similar proof for the DNLS equation.

It is also interesting that in the case $d\in (d_+,\infty)$, we find that the complex spectral band in the Floquet spectrum can only intersect the imaginary axis inside the spectral gap or in the interior spectral band but not in the exterior spectral bands. It would be interesting to see why this is the case analytically. 

In summary, analysis of the KN spectral problem for the periodic standing waves in the DNLS equation is open for further study.

\vspace{0.25cm}

{\bf Acknowledgements.} This work was supported in part by the National
Natural Science Foundation of China (No. 11971103).

\end{document}